\documentclass{theoretics}

\title{Algorithmizing the Multiplicity Schwartz-Zippel Lemma}
\ThCSshorttitle{Algorithmizing the Multiplicity Schwartz-Zippel Lemma}
\ThCSyear{2025}
\ThCSarticlenum{29}
\ThCSreceived{Mar 29, 2024}
\ThCSrevised{Apr 22, 2025}
\ThCSaccepted{May 23, 2025}
\ThCSpublished{Dec 18, 2025}
\ThCSdoicreatedtrue

\ThCSauthor[TTI]{Siddharth Bhandari}{siddharth@ttic.edu}[https://orcid.org/0000-0003-3481-6078]
\ThCSauthor[TIFR]{Prahladh Harsha}{prahladh@tifr.res.in}[https://orcid.org/0000-0002-2739-5642]
\ThCSauthor[TIFR]{Mrinal Kumar}{mrinal.kumar@tifr.res.in}[https://orcid.org/0000-0002-6430-0219]
\ThCSauthor[TIFR]{Ashutosh Shankar}{ashutosh.shankar@tifr.res.in}[https://orcid.org/0009-0003-9308-4054]
\ThCSaffil[TTI]{Toyota Technological Institute, Chicago, USA}
\ThCSaffil[TIFR]{Tata Institute of Fundamental Research, Mumbai, India}

\ThCSthanks{%
	A preliminary version of this article appeared at SODA 2023~\cite{BhandariHKS2023}. %
	Research of the first, second and fourth authors supported in part by the Department of Atomic Energy, Government of India, under project 12-R\&D-TFR-5.01-0500. Work done when the first author was a PhD student at TIFR where his research was supported in part by the Google PhD Fellowship, and the third author was at IIT Bombay. Research of the second author supported by the Swarnajayanti Fellowship.}

\ThCSshortnames{S. Bhandari, P. Harsha, M. Kumar, A. Shankar}

\ThCSkeywords{Error-correcting codes, coding theory, algebraic coding theory, information theory}

\usepackage{thmtools}

\usepackage{stmaryrd}
\usepackage{xfrac}
\usepackage{mathtools}
\usepackage{nicefrac}
\usepackage{float}
\let\oldnl\nl% Store \nl in \oldnl
\newcommand{\nonl}{\renewcommand{\nl}{\let\nl\oldnl}}% Remove line number for one line

\usepackage{tikz}
\usepackage{float}
\allowdisplaybreaks
\usepackage{csquotes}
\usepackage[T2A,T1]{fontenc}
\usepackage[russian,english]{babel}
\usepackage{nth}
\usepackage{intcalc}
\usepackage{etoolbox}
\usepackage{xstring}

\let\N\Natural

\let\Z\Integer
\let\F\Field

\def\argmax{\operatornamewithlimits{arg\,max}}

\def\E{\operatornamewithlimits{\mathbb{E}}}			% Expectation: $\E[X]$ (like \Pr)
\def\poly{\operatorname{poly}}

\def\set#1{\left\{ #1 \right\}}
\def\abs#1{\left| #1 \right|}

\def\inbrace#1{\left\{ #1 \right\}}

\def\setdef#1#2{\inbrace{ #1 \ \colon \ #2}}      % E.g: \setdef{x}{f(x) = 0}

\def\veca{\mathbf{a}}

\def\vece{\mathbf{e}}

\def\veci{\mathbf{i}}

\def\vecw{\mathbf{w}}
\def\vecx{\mathbf{x}}
\def\vecy{\mathbf{y}}
\def\vecz{\mathbf{z}}

\usetikzlibrary{decorations.pathreplacing,arrows, calc, shapes, chains}

\renewcommand{\tilde}[1]{\widetilde{#1}}
\newtheorem*{combSZ}{Unique Decoding (Combinatorial Statement)}
\newtheorem*{mainthm}{Main result}

\newcommand{\hpartial}{{\mathchar'26\mkern-10mu \partial}}

\renewcommand{\E}{\mathbb{E}}

\newcommand{\ip}[1]{\left\langle #1 \right\rangle}

\newcommand{\enc}{{\mathsf{Enc}}}
\newcommand{\encs}{{\mathsf{Enc}^{(s)}}}
\newcommand{\encst}{{\mathsf{Enc}_{T^m}^{(s)}}}

\newcommand{\calR}{\mathcal{R}}
\newcommand{\calL}{\mathcal{L}}

\let\epsilon\varepsilon
\newcommand{\coeff}{\operatorname{Coeff}}
\newcommand{\mult}{\operatorname{mult}}

\newcommand{\Del}{\Delta}
\newcommand{\Delm}{\Delta_{\mathsf{mult}}}
\newcommand{\ow}{\widetilde{w}}

\newcommand{\vy}{\vecy}
\newcommand{\vz}{\vecz}

\newcommand{\vx}{\vecx}
\newcommand{\vi}{\veci}

\newcommand{\ve}{\vece}

\newcommand{\va}{\veca}

\newcommand{\vw}{\vecw}

\newcommand{\sbar}{\mathbf{s}}
\newcommand{\thetabar}{{\boldsymbol{\theta}}}

\DontPrintSemicolon

\usepackage[capitalise,noabbrev,nameinlink]{cleveref}

        \SetNlSkip{0.3em}
        \SetKwComment{Comment}{$\triangleright$\ }{}

        \newlength\mylen
     \newcommand\myinput[1]{%
     \settowidth\mylen{\KwIn{}}%
     \setlength\hangindent{\mylen}%
     \nonl \hspace*{\mylen}#1\\}
  
\begin{document}
	\maketitle
	
	\begin{abstract}
		The multiplicity Schwartz-Zippel lemma asserts that over a field, a
		low-degree polynomial cannot vanish with high multiplicity very
		often on a sufficiently large product set. Since its discovery in a
		work of Dvir, Kopparty, Saraf and Sudan~\cite{DvirKSS2013}, the
		lemma has found numerous applications in both math and computer
		science; in particular, in the definition and properties of
		multiplicity codes by Kopparty, Saraf and Yekhanin~\cite{KoppartySY2014}.
		
		In this work, we show how to algorithmize the multiplicity
		Schwartz-Zippel lemma for arbitrary product sets over any field. In
		other words, we give an efficient algorithm for unique decoding of
		multivariate multiplicity codes from half their minimum distance on
		arbitrary product sets over any field and any constant multiplicity parameter. Previously, such an
		algorithm was known either when the underlying product set had a nice
		algebraic structure (for instance, was a subfield)~\cite{Kopparty2015} or when the underlying field had large (or zero)
		characteristic, the multiplicity parameter was sufficiently large
		and the multiplicity code had distance bounded away from $1$~\cite{BhandariHKS2021-mgrid}. 
		
		Our algorithm builds upon a result of Kim \& Kopparty~\cite{KimK2017} who gave an algorithmic version of the
		Schwartz-Zippel lemma (without multiplicities) or equivalently, an
		efficient algorithm for unique decoding of Reed-Muller codes over
		arbitrary product sets. We introduce a refined notion of distance based on the multiplicity Schwartz-Zippel lemma and design a unique decoding algorithm
		for this distance measure. On the way, we give an alternative analysis of
		Forney's classical generalized minimum distance decoder that might
		be of independent interest.
		
	\end{abstract}
	
	\section{Introduction}
	The \emph{degree-mantra} states that any univariate non-zero
	polynomial $P \in \F[x]$ of degree $d$, where $\F$ is a field, has at
	most $d$ zeros even including multiplicities. A generalization of this
	basic maxim to the multivariate setting is the well-known
	\emph{Schwartz-Zippel} Lemma, also sometimes referred to as the
	\emph{Polynomial-Identity} Lemma, (due to Ore~\cite{Ore1922},
	Schwartz~\cite{Schwartz1980}, Zippel~\cite{Zippel1979} and DeMillo \&
	Lipton~\cite{DemilloL1978}). This Lemma states that if $\F$ is a
	field, and $P \in \F[x_1, x_2, \ldots, x_m]$ is a \emph{non-zero}
	polynomial of degree $d$, and $T\subseteq \F$ is an arbitrary finite
	subset of $\F$, then the number of points on the grid $T^m$ where $P$
	is zero is upper bounded by $d|T|^{m-1}$. A generalization of this lemma that incorporates the
	\emph{multiplicity} aspect of the (univariate) degree-mantra was proved
	by Dvir, Kopparty, Saraf and Sudan~\cite{DvirKSS2013}. This
	multiplicity Schwartz-Zippel Lemma (henceforth, referred to as the
	multiplicity SZ lemma for brevity) states that the number
	of points on the grid $T^m$ where $f$ is zero with
	\emph{multiplicity}\footnote{This means that all the partial
		derivatives of $P$ of order at most $s-1$ are zero at this
		point. See \cref{sec:mult} for a formal definition.} at least
	$s$ is upper bounded by $\frac{d|T|^{m-1}}{s}$.\footnote{This bound is
		only interesting when $|T| > d/s$ so that $\frac{d|T|^{m-1}}{s}$ is
		less than the trivial bound of $|T|^{m}$.}
	
	This innately basic statement about low degree polynomials is crucial to many applications in both computer science and math, e.g. see~\cite{Guth2016}. Of particular interest to this work is the significance of this lemma and its variants to algebraic error-correcting codes that we now discuss in more detail.   
	
	\subsection{The coding theoretic context and main theorem}
	SZ lemma and its variants play a fundamental role in the study of error-correcting codes. The univariate version (referred to as the \emph{degree mantra} in the first paragraph) implies a lower bound on the distance of Reed-Solomon codes whereas for the general multivariate setting, the SZ lemma gives the distance of Reed-Muller codes. In its multiplicity avatar (both the univariate and multivariate settings), the lemma implies a lower bound on the distance of the so-called \emph{multiplicity codes}. Informally, multiplicity codes are generalizations of Reed-Muller codes where in addition to the evaluation of a message polynomial at every point on a sufficiently large grid, we also have the evaluations of all partial derivatives of the polynomial up to some order at every point on a grid. Thus, the number of agreements between two distinct codewords of this linear code is exactly  the number of points on a grid where a low-degree polynomial vanishes with high multiplicity, and thus, is precisely captured by the multiplicity SZ lemma. 
	
	A natural algorithmic question in the context of error-correcting codes is that of decoding: given a received word that is sufficiently close to a codeword in Hamming distance, can we efficiently find the codeword (or the message corresponding to the closest codeword). For Reed-Solomon codes, this decoding question is equivalent to the question of designing efficient algorithms for recovering a univariate polynomial from its evaluations, with the additional complication that some of the evaluations might be erroneous. For Reed-Muller codes, this is the question of multivariate polynomial interpolation with errors over grids, and for multiplicity codes, this is the question of recovering a multivariate polynomial from its evaluations and evaluations of all its partial derivatives up to some order on  a grid, where some of the data is erroneous.  If the number of coordinates with errors is at most half the minimum distance of the code, then we know that there is a \emph{unique} codeword close to any such received word, and what the minimum distance actually is implied by the SZ lemma or its variants. In particular, the multiplicity SZ lemma can be restated as the following combinatorial statement about unique decodability of multiplicity codes. 
	\begin{combSZ}\label{thm:question}
		Let $\F$ be a field, and $T \subseteq \F$, $d$ the degree parameter,
		$m$ the dimension and $s$ the multiplicity parameter be as
		above. Given an arbitrary function
		$f\colon T^{m} \to \F_{<s}[z_1,\ldots,z_m]$, where $\F_{<s}[z_1,\ldots,z_m]$ denotes the set of all polynomials of degree less than $s$ in the variables $z_1,z_2,\ldots,z_m$, there exists \emph{at most
			one} polynomial $P \in \F[x_1, \ldots, x_m]$ of degree at most $d$
		such that the function $\enc^{(s)}(P)\colon T^{m} \to \F_{<s}[z_1,\ldots,z_m]$
		defined as
		\[
		\enc^{(s)}(P)(\va) := P(\va + \vz) \mod \langle \vz
		\rangle^s 
		\]
		differs from $f$ on less than $\frac12\left(1-\frac{d}{s|T|}\right)$ fraction of points on $T^m$. 
		Here $\langle \vz \rangle^s $ denotes the ideal in $\F[\vz]$ generated by monomials of total degree $s$.
	\end{combSZ}
	\begin{remark}
		We note that the encoding function $\enc^{(s)}(P)(\va)$ is sometimes also  defined as being given by the evaluation of all partial derivatives of $P$ of order at most $s-1$ on input $\va$. These definitions are equivalent as can be seen by looking at the Taylor expansion of $P(\va + \vz)$ and truncating the series at monomials of degrees less than $s$ in $\vz$.  
	\end{remark}
	Thus, the unique decoding question for multiplicity codes is equivalent to asking whether there is an algorithmic equivalent of the multiplicity SZ lemma. More precisely, given a function $f$, can one \emph{efficiently} find the (unique)
	polynomial $P$ such that
	$\Delta(f,\enc^{(s)}(P)) <\frac12\left(1-\frac{d}{s|T|}\right)$ (if one
	exists). 
	
	Given the central role that these polynomial evaluation codes play in
	coding theory and in complexity theory, it is not surprising to note
	that this question of efficient unique decoding of Reed-Solomon codes
	and Reed-Muller codes in particular has been widely investigated. For
	Reed-Solomon codes, we know multiple algorithms for this task,
	starting with the result of Peterson from the 50s~\cite{Peterson1960}
	and the subsequent algorithms of Berlekamp and Massey
	\cite{Berlekamp,Massey1969} and Welch and Berlekamp
	\cite{WelchB1986}. However, the situation is a bit more complicated
	for Reed-Muller and multivariate multiplicity codes.
	
	Till recently, all known efficient decoding algorithms for the multivariate setting (both Reed-Muller and larger order multiplicity codes),  worked only when
	the underlying set $T$ had a nice algebraic structure (e.g., $T = \F$)
	or when the degree $d$ was very small (cf, the Reed-Muller
	list-decoding algorithm of Sudan~\cite{Sudan1997} and its multiplicity
	variant due to Guruswami \& Sudan~\cite{GuruswamiS1999}). In particular, even for Reed-Muller codes (equivalently multiplicity codes with multiplicity parameter $s=1$), where the unique decoding question corresponds to an algorithmic version of the standard SZ lemma, no \emph{efficient} unique decoding algorithm was known. In fact, the problem was open even in the bivariate setting! This seems a bit surprising since the distance of these codes just depends on the non-vanishing of polynomials on arbitrary grids, whereas the decoding algorithms appear to crucially use some very specific algebraic properties of the grid. Even beyond the immediate connection to questions in coding theory, this seems like a very natural algorithmic question in computational algebra that represents a gap in our understanding of a very fundamental property of low-degree polynomials.  
	
	For Reed-Muller codes, this question was resolved in a beautiful work of Kim and Kopparty~\cite{KimK2017} who gave an efficient algorithm for this problem of unique decoding Reed-Muller codes on arbitrary grids. Their algorithm was essentially an algorithmic version of the standard induction-based proof of the SZ lemma. However, the algorithm of Kim and Kopparty does not seem to generalize for the case of higher multiplicity ($s > 1$), and indeed, this problem is mentioned as an open problem of interest in~\cite{KimK2017}. 
	%In his survey on multiplicity codes~\cite{Kopparty15}, Kopparty asks whether the list decoding properties of multiplicity codes ().
	
	Since the work of Kim and Kopparty, Bhandari, Harsha, Kumar and Sudan~\cite{BhandariHKS2021-mgrid} made some
	partial progress towards this problem for the case of $s>1$ and
	$m >1$. In particular, they designed an efficient decoding algorithm
	albeit requiring the following conditions to be satisfied: 
	\begin{itemize}
		\item  the field $\F$ has characteristic zero or is larger than the degree $d$
		\item the distance of the code is constant, and
		\item the multiplicity parameter $s$ is sufficiently large (in terms of the dimension $m$)
	\end{itemize}
	Under these special conditions, they in
	fact obtained a list-decoding algorithm. Yet, the original algorithmic
	challenge of obtaining an efficient unique decoding for \emph{all}
	dimensions and \emph{all} multiplicity parameters remained
	unresolved. In particular, even unique
	decoding of bivariate multiplicity codes with multiplicity two from
	half their minimum distance was not known over arbitrary product
	sets over any field.
	
	Our main result in this work is a generalization of the result of Kim
	and Kopparty  to \emph{all} constant multiplicities or equivalently, an algorithmic version of the multiplicity SZ lemma (for constant multiplicity parameter). More formally, we prove the following theorem. 
	
	\begin{mainthm}\label{thm: main intro}
		Let $s, d, m \in \N$ and $T \subseteq \F$ of size $n$ be such that $d <
		sn$. Then, there is a deterministic algorithm that runs in time
		$(sn)^{O(s+m)}\cdot  \binom{s - 1 + m}{m}$ and on
		input $f\colon T^m \to \mathbb{F}_{<s}[\vx]$ outputs the unique polynomial
		$P \in \mathbb{F}_{\leq d}[\vx]$ (if such  a polynomial exists) such that $\enc^{(s)}(P)$
		differs from $f$ on less than $\frac12\left(1-\frac{d}{s|T|}\right)$ fraction of points on $T^m$.   
	\end{mainthm}
	
	We remark that our algorithm is not polynomial-time in the input size for all settings of the multiplicity parameter $s$; our algorithm runs in time $(sn)^{O(s+m)}\cdot  \binom{s - 1 + m}{m}$ while the input-size is $n^{m}\cdot  \binom{s - 1 + m}{m}$. That said, it is polynomial when $s$ is a constant, the typical setting for coding-theoretic applications. While we are chiefly interested in the constant multiplicity parameter regime, we note that our algorithm is polynomial for certain super-constant settings of $s$ also, which might be useful in some applications. We conjecture that it can be made to run in polynomial time for all settings of $s$. 
	
	The main theorem here differs from the results of~\cite{BhandariHKS2021-mgrid} in the following sense: our main theorem gives a polynomial-time algorithm for unique decoding multiplicity codes for multiplicity parameter $s$ being any arbitrary constant and over any underlying field, whereas~\cite{BhandariHKS2021-mgrid} give an efficient algorithm for list decoding multiplicity codes up to the list decoding capacity, provided that the multiplicity parameter is sufficiently large (given by the dimension $m$) and the underlying field has large or zero characteristic. 
	
	\subsection{Alternative analysis of Forney's GMD Decoding}
	
	Forney designed the Generalized Minimum Distance (GMD) decoding to
	decode concatenated codes from half its minimum distance~\cite{Forney1965, Forney1966}. A key step in
	our algorithm (as in Kim and Kopparty's Reed-Muller decoder) is
	reminiscent of Forney's GMD decoding. Forney analysed the GMD decoding
	using a convexity argument, which in modern presentations is usually
	expressed as a probabilistic or a random threshold argument.  Kim and Kopparty used a direct adaptation
	of this random threshold argument for their Reed-Muller decoder. This
	argument unfortunately fails in our setting. To get
	around this, we first give an alternative analysis of Forney's GMD
	decoding using a matching argument, which we explain in detail in the next section. The alternative
	analysis of Forney's GMD decoding is given in detail in \cref{sec:forney}.

	\subsection{Further discussion and open problems}
	We conclude the introduction with some open problems.
	\begin{description}
		\item [Faster algorithms.] Observe that the running time of the main result  as stated above depends polynomially on $(sn)^{s+m}$. Since the input to the decoding problem has size $\left(n^m\cdot \binom{s - 1 + m}{m} \right)$, strictly speaking, the running time of the algorithm in the main theorem is not polynomially bounded in the input size for all choices of the parameters $n, s, m$. It would be extremely interesting to obtain an algorithm for this problem which runs in polynomial time in the input size for all regimes of parameters. 
		
		\item[List-decoding. ] Another natural question in the context of the main result in this paper is to obtain efficient algorithms for decoding multiplicity codes on arbitrary grids when the amount of error exceeds half the minimum distance of the code. The results in~\cite{BhandariHKS2021-mgrid} provide such algorithms when $s$ is sufficiently large, $m$ is a constant and we are over fields of sufficiently large (or zero) characteristic. However, when $s$ is small, for instance, even when $s= 1$ and $m=2$, we do not have any such list decoding results. In fact, for small $s$, we do not even have a good understanding of the combinatorial limits of list decoding of these (Reed-Muller/Multiplicity) codes on arbitrary grids.

		\item[Polynomial-method-based algorithms. ] It would also be interesting to understand if the results in this paper and those in~\cite{KimK2017} can be proved via a more direct application of the polynomial method. By this we mean a decoding algorithm along the lines of the classical Welch-Berlekamp~\cite{WelchB1986} decoding algorithm as presented by Gemmell and Sudan~\cite{GemmellS1992} or the generalizations used for list-decoding: for instance, as in the work of Guruswami and Sudan~\cite{GuruswamiS1999}.
		Note that for large $s$, constant $m$, and fields of sufficiently large or zero characteristic, the algorithm in~\cite{BhandariHKS2021-mgrid} is indeed directly based on a clean application of the polynomial method. It would be aesthetically appealing to have an analogous algorithm for the case of small $s$ (and perhaps over all fields).
	\end{description}
	
	\subsection{Organization}
	The rest of the paper is organized as follows. 
	
	% In \cref{sec: overview}, we give an overview of the main ideas in
	% our algorithm and discuss how they relate to the algorithm of Kim
	% and Kopparty~\cite{KimK2017} for decoding Reed-Muller codes.
	We begin with an overview of our algorithm in \cref{sec: overview}. In \cref{sec: mult prelims}, we discuss the necessary preliminaries including the definition and properties of multiplicity codes and describe  a fine-grained notion of distance for multiplicity codes that plays a crucial role in our proofs. We build up the necessary machinery for the analysis of the bivariate decoder in \cref{sec:BWmultiplicity} and \cref{sec:wumd} where we describe and analyse efficient decoding algorithms for univariate multiplicity codes with varying multiplicities and weighted univariate multiplicity codes respectively. In \cref{sec:bivariate}, we put it all together and describe our algorithm for decoding bivariate multiplicity codes, and give its proof of correctness. Using a suitable induction on the number of variables, we generalize the algorithm for the bivariate case to the multivariate case in \cref{sec: multivariate algorithm}. An alternative analysis of Forney's generalized minimum distance decoder, which inspires our analysis of our algorithm, is discussed in \cref{sec:forney}.

	\section{Overview of algorithm}\label{sec: overview}
	
	As one would expect, our multivariate multiplicity code decoder for
	$s> 1$ is a generalization of the Kim-Kopparty decoder for
	Reed-Muller codes (i.e., the $s=1$ case). However, a straightforward
	generalization does not seem to work and several subtle issues
	arise. Some of these issues are definitional and conceptual; for
	instance, one needs to work with a finer notion of distance while
	others are more technical, for instance the need for an alternative analysis of
	Forney's GMD decoding.

	To explain these issues and discuss how we circumvent them, we first recall
	the Kim-Kopparty decoder and then mention how we generalize it to
	the $s > 1$ setting, highlighting the various issues that arise along
	the way.
	
	\subsection{An overview of the Kim-Kopparty Reed-Muller decoder}
	
	We begin by recalling the Kim-Kopparty decoder for Reed-Muller codes
	on arbitrary grids. Their algorithm is, in some sense, an
	``algorithmization'' of the standard inductive proof of the classical
	Schwartz-Zippel lemma (without any multiplicities). For simplicity, we
	will focus on the bivariate case ($m=2$) initially.
	
	The setting is as follows: we are given a received word
	$f\colon T^2 \to \mathbb{F}$ and have to find the unique polynomial
	$P(x, y) \in \mathbb{F}[x, y]$ (if one exists) such that
	$\Delta(f,P) < \frac{1}{2} n^2 (1-\frac{d}{n})$, where $\Delta$ is
	the standard Hamming distance and $n=|T|$.
	
	The high-level approach is to write $P$ as
	$\sum_{i=0}^d P_i(x) y^{d-i}$ and iteratively recover the univariate
	polynomials $P_0, P_1, \ldots, P_d$ in this order. The goal of the
	$\ell^{th}$ iteration is to recover $P_\ell(x)$ assuming we have
	correctly recovered $P_0, P_1,\ldots, P_{\ell-1}$ in the previous
	iterations. For the $\ell^{th}$ iteration, we work with $f_{\ell}\colon T^2 \to \F$
	given by
	$f_{\ell}(a,b) := f(a,b) - \sum_{i = 0}^{\ell-1}P_i(a)\cdot b^{d-i}$ as
	the received word. Clearly, the Hamming distance between $f_{\ell}$
	and the evaluation of the bivariate polynomial
	$Q_{\ell}(x,y) = \sum_{i = \ell}^d P_i(x)y^{d-i}$ on the grid is exactly
	$\Delta(f,P)$ which is at most $\frac{1}{2} n^2 (1-\frac{d}{n})$.

	Very naturally, we view $f_{\ell}$ as values written on a two-dimensional grid
	$T \times T$, with the columns indexed by $x = a \in T$ and the rows
	indexed by $y = b \in T$, and each entry on a grid point of
	$T \times T$ is an element of $\F$. As a first step of the algorithm,
	for every $a \in T$, a Reed-Solomon decoder is used to find a
	univariate polynomial $G^{(a)}(y)$ of degree at most $d - \ell$ such that the tuple
	of evaluations of this polynomial on the set $T$ is close to the
	restriction of $f_{\ell}$ on the column $x = a$. For columns where $Q_{\ell}$ and $f_{\ell}$
	have a high agreement, $G^{(a)}(y)$ would be equal to $Q_{\ell}(a, y)$ and
	hence, the leading coefficient of $G^{(a)}(y)$ must equal
	$P_{\ell}(a)$. A natural suggestion here would be to collect the
	leading coefficients of the polynomials $G^{(a)}(y)$'s, namely $g(a) := \coeff_{y^{d-\ell}}(G^{(a)}(y))$  and try and run another
	Reed-Solomon decoder on $g\colon T\to \F$ to get $P_{\ell}(x)$. Unfortunately, with
	$\Delta(f,P) < \frac{1}{2} n^2 (1-\frac{d}{n})$, it can be shown that
	a lot of these polynomials $G^{(a)}$ could be incorrect and hence, the
	number of errors in the received word for this second step
	Reed-Solomon decoding is too large for it to correctly output
	$P_{\ell}$.
	
	The key idea of Kim and Kopparty is to observe that at the end of the
	first step, not only do we have the decoded column polynomials
	$G^{(a)}(y)$ but we also have hitherto unused  information about
	the function $f_{\ell}$, namely, the Hamming distance between the
	received word (the restriction of $f_{\ell}$) on a column $x = a$ and
	the evaluation of $G^{(a)}(y)$.
	Formally, for every $a \in T$, they associate a
	weight
	\begin{align}
		w(a) := \min\left\{{\Delta\left(f_{\ell}(a, \cdot), G^{(a)}\right)},\frac{n
			-(d-\ell)}{2}  \right\}\, .\label{eq:KKweight}
	\end{align}
	Note that if $\Delta\left(f_{\ell}(a, \cdot), G^{(a)}\right) > \frac{n
		-(d-\ell)}{2}$, then there is no guarantee that $G^{(a)}(y)$ is the
	closest codeword to $f_{\ell}(a,y)$ and hence the minimum is taken in
	the above expression to cap the distance to $\frac{n
		-(d-\ell)}{2}$ (this plays a technical reason in the proof, but let
	us ignore it for now). 
	Based on these weights, they define a modified distance between the
	pair $(g,w)\colon T \to \F \times \Z_{\geq 0}$, which they refer to as
	a ``fractional word'',  and a regular word $h \colon T \to \F$ as
	follows:
	\begin{align}
		\Delta((g,w), h) := \sum_{a \in A_0(g,h)} (n-(d-\ell) - w(a)) + \sum_{a
			\in A_1(g,h)} w(a) \;,\label{eq:KKdistance}
	\end{align}
	where $T= A_0(g,h) \cup A_1(g,h)$ is the following partition of $T$
	\begin{align*}
		A_1(g,h) &:= \{ a \in T \colon g(a) = h(a) \}\;,
		\tag*{\text{(agreement points)}}\\
		A_0(g,h) &:= \{ a \in T \colon g(a) \neq h(a)\}\;. \tag*{\text{(disagreement points)}}
	\end{align*}
	We note that the above is a scaled version of the modified distance
	used by Kim and Kopparty. We find it convenient to express it in the
	above equivalent form for the purpose of this presentation as this
	alternate form generalizes to the $s>1$ setting more naturally.
	
	They then prove that this modified distance satisfies the following
	two important properties.
	\begin{itemize}
		\item The fractional word $(g,w)$ and the ``correct'' polynomial
		$P_\ell$ satisfy
		\begin{align}\Delta((g,w),P_\ell) &\leq \Delta(f,P) < \frac12
			n^2\left(1-\frac{d}{n}\right)\;.\label{eq:KKprop1}
		\end{align}
		\item For any two distinct polynomials $Q_\ell$ and $R_\ell$ of
		degree $\ell$, we have
		\begin{align}
			\Delta((g,w), Q_\ell) + \Delta((g,w),R_\ell) &\geq
			(n-(d-\ell))\cdot \Delta(Q_\ell,R_\ell)\notag\\
			&\geq
			(n-(d-\ell))
			\cdot
			(n-\ell)\notag\\
			&\geq n^2\left(1
			-\frac{d}{n}\right)\;. \label{eq:KKprop2}
		\end{align}
	\end{itemize}
	These two properties imply that the pair $(g,w)$ uniquely determines
	the polynomial $P_\ell$ since $P_\ell$ satisfies
	$\Delta((g,w),P_\ell) <  \frac12
	n^2\left(1-\frac{d}{n}\right)$ (by the first property) while  every other polynomial $Q_\ell$
	satisfies $\Delta((g,w), Q_\ell) \geq n^2\left(1
	-\frac{d}{n}\right) - \Delta((g,w),P_\ell) > \frac12
	n^2\left(1-\frac{d}{n}\right)$ (by the second property).
	
	We are still left with the problem
	of \emph{efficiently} determining $P_\ell$ from the pair $(g,w)$. We will
	defer that discussion to later, but first show how the above steps
	can be generalized to the $s>1$ setting. 
	
	\subsection{Generalizing the Kim-Kopparty decoder to \texorpdfstring{$s>1$}{s > 1}}
	To begin with, we will need to work with a more fine-grained notion of
	distance which incorporates the multiplicity information. For two
	functions $f, g \colon T^m \to \F_{<s}[\vz]$ and a point $\va \in
	T^m$, the Hamming distance $\Delta(f(\va), g(\va))$ measures if the polynomials $f(\va)$ and
	$g(\va)$ are different. We will consider a refined notion of $\Delta_{\mathsf{mult}}^{(s)}(f(\va),g(\va))$, which measures the
	``multiplicity distance'' at the point $\va$ (here $s$ is the
	multiplicity parameter). Formally, 
	\[ \Delta_{\mathsf{mult}}^{(s)}(f(\va),g(\va)) := (s -
	d^{(s)}_{\min}(f(\va) - g(\va)))\;,\]
	where $d^{(s)}_{\min}(R)$ is defined as follows for any polynomial
	$R$: $d^{(s)}_{\min}(R)$ is the minimum of $s$ and the degree of the
	minimum degree monomial with a non-zero coefficient in $R$. Note that if $R$ is
	identically zero, $d^{(s)}_{\min}(R)$ is set to $s$. Recall that here $f(\va)$ and $g(\va)
	$ are both polynomials in $\vz$ of degree strictly less than $s$. The multiplicity distance between the functions $f$ and $g$ is defined
	to be the sum of the above quantity over all $\va \in T^m$. More
	precisely,
	\[ \Delta_{\mathsf{mult}}^{(s)}(f,g) := \sum_{\va \in T^m}
	\Delta_{\mathsf{mult}}^{(s)}(f(\va),g(\va))\;.
	\]
	
	Observe that for $s=1$, this matches with the usual notion of Hamming
	distance. Furthermore, this fine-grained distance lower-bounds the
	Hamming distance as follows:
	\[ \Delta_{\mathsf{mult}}^{(s)}(f,g) \leq s \cdot \Delta(f,g)
	\;. \]
	
	The multiplicity SZ lemma, in terms of this fine-grained
	multiplicity distance, states that for two distinct $m$-variate degree $d$
	polynomials $P,Q$, we have
	\[\Delta_{\mathsf{mult}}^{(s)} (\enc^{(s)}(P),\enc^{(s)}(Q))\geq n^{m} \cdot \left(s - \frac{d}{n}\right).\]
	
	The Kim-Kopparty decoder for Reed-Muller codes on a grid required as a basic
	primitive a Reed-Solomon decoder. Correspondingly, our algorithm will
	require as a basic primitive a univariate multiplicity code
	decoder. Such a decoder was first obtained by
	Nielsen~\cite{Nielsen2001}. Later on in the algorithm, we will
	actually need a  generalization of this decoder which can
	handle the case when the multiplicity parameters at different
	evaluation points are not necessarily the same.  Such a decoder can be obtained by a
	suitable modification of the standard Welch-Berlekamp decoder for
	Reed-Solomon codes. This modification is presented in
	\cref{sec:BWmultiplicity} (Algorithm~\ref{alg:gen mult decoder} specifically). 
	
	We are now ready to generalize the Kim-Kopparty decoder to the $s>1$
	setting. To keep things simple, let us focus on the $s=2$ and $m=2$
	setting. The problem stated in terms of the multiplicity-distance is
	as follows. We are given a received word
	$f\colon T^2 \to \F_{<2}[u,v]$ and have to find the unique polynomial
	$P(x, y) \in \mathbb{F}[x, y]$ (if one exists) such that
	$\Delta_{\sf{mult}}^{(2)}(f,\enc^{(2)}(P)) < \frac{1}{2} n^2
	(2-\frac{d}{n})$, where $n=|T|$.

	As before we write $P= \sum_{i=0}^d P_i(x) y^{d-i}$ and recover the
	polynomials $P_0,\ldots, P_d$ in successive iterations. Let us assume
	we are in the $\ell^{th}$ iteration where we have to recover $P_\ell$ and
	have already correctly recovered $P_0,\ldots, P_{\ell-1}$ in the previous
	iterations. Similar to the $s=1$ setting, we consider the function $f_{\ell}\colon T^2 \to \F_{<2}[u,v]$
	given by
	$f_{\ell}(a,b) := f(a,b) - \enc^{(2)}\left(\sum_{i = 0}^{\ell-1}P_i(x)\cdot y^{d-i}\right)$ as
	the received word. Clearly, the distance between $f_{\ell}$
	and $\enc^{(2)}\left(\sum_{i = \ell}^d P_i(x)y^{d-i}\right)$ is exactly
	$\Delta_{\sf{mult}}^{(2)}(f,\enc^{(2)}(P))$ which is at most $\frac{1}{2}
	n^2 (2-\frac{d}{n})$.

	As before, we view $f_{\ell}$ as values written on a two-dimensional
	grid $T \times T$, with the columns indexed by $x = a \in T$ and the
	rows indexed by $y = b \in T$, and each entry on a grid point of
	$T \times T$ is an element of $\F_{<2}[u,v]$. It will be convenient to
	view this degree-one polynomial $f_\ell(a,b)$ at the grid point
	$(a,b)$ as $f_\ell^{(0)}(a,b)(v) + u \cdot f_\ell^{(1)}(a,b)$ where
	$f_\ell^{(0)}(a,b) \in \F_{<2}[v]$ and
	$f_\ell^{(1)}(a,b) \in F_{<1}[v] \equiv \F$. For every $a \in T$, we look
	at the restriction of the functions $f^{(0)}_\ell$ and $f^{(1)}_\ell$
	to the column $x=a$. We could use a Reed-Solomon decoder to find the
	univariate polynomial $G^{(a,1)}(y)$ of degree at most $d - \ell$
	such that the tuple of evaluations of this polynomial on the set $T$
	is close to the restriction of $f^{(1)}_{\ell}$ on the column $x =
	a$. Similarly, we could use the univariate multiplicity code decoder
	(mentioned above) to
	find the univariate polynomial $G^{(a,0)}(y)$ of degree at most
	$d - \ell$ such that encoding of this polynomial is close to the
	restriction of $f^{(0)}_{\ell}$ on the column $x = a$. We can for each
	$a$, define $g(a) \in \F_{<2}[u]$ to be the polynomial
	$g(a):=\coeff_{y^{d-\ell}}(G^{(a,0)}) + u\cdot
	\coeff_{y^{d-\ell}}(G^{(a,1)})$. 
	
	\begin{figure}
		\centering
		\includegraphics{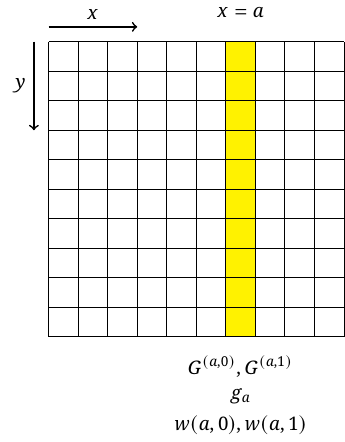}
		\caption{Bivariate algorithm, visualised}
		\label{fig:bivariate grid}
	\end{figure}
		
	We also have the corresponding weight
	functions (similar to \eqref{eq:KKweight}) which measure how close the encodings of the polynomials $G^{(a,0)}(y)$
	and $G^{(a,1)}(y)$ are close to the corresponding functions $f_\ell^{(0)}$
	and $f_\ell^{(1)}$ respectively. Namely, 
	\begin{align*}
		w(a,0) &:=
		\min\left\{{\Delta_{\mathsf{mult}}^{(2)}\left(f^{(0)}_{\ell}(a,
			\cdot) , \enc^{(2)}(G^{(a,0)}) \right)},
		{\frac{2n-(d-\ell)}{2}}  \right\}\;,\\
		w(a,1) &:=
		\min\left\{{\Delta_{\mathsf{mult}}^{(1)}\left(f^{(1)}_{\ell}(a, \cdot), \enc^{(1)}(G^{(a,1)})
			\right)}, {\frac{n-(d-\ell)}{2}}
		\right\}\;.
	\end{align*}
	The quantities $2n-(d-\ell)$ and $n-(d-\ell)$ are the distances of
	univariate multiplicity codes of degree $(d-\ell)$ and multiplicity
	parameters $2$ and $1$ respectively. 
	We will refer to these pairs of weight functions as $\vw = (w(\cdot,0),
	w(\cdot,1))$. As in the $s=1$ setting, we will view the pair
	$(g,\vw)$ as a ``fractional word'' and would like to define the distance between the fractional word $(g,\vw)$ and $h$. In the $s=1$ Reed-Muller setting (\eqref{eq:KKdistance}), the set $T$ was partitioned into two sets: the set $A_1(g,h)$ of agreement points and the set $A_0(g,h)$ of disagreement points. However, for the $s=2$ case, agreement/disagreement come in more flavours. We have points in $T$ in which $g$ and $h$ agree in both the evaluation and derivative, points in which they agree only in the evaluation but not the derivative and finally points where they disagree even on the evaluation. Based on this, we have the following partition of the set $T= A_0(g,h) \cup A_1(g,h) \cup A_2(g,h)$:
	\begin{align*}
		A_2(g,h) &:= \left\{ a \in T \colon g(a) \equiv h(a) \right\}\;,
		\tag*{\text{(agreement points with multiplicity 2)}}\\
		A_1(g,h) &:=  \left\{ a \in T \colon g(a) \equiv h(a) \mod \langle u
		\rangle \right\} \setminus A_2(g,h)\;, \tag*{\text{(agreement points with multiplicity 1)}}\\
		A_0(g,h) &:=  T \setminus \left( A_2(g,h)    \cup
		A_1(g,h)\right)\;.  \tag*{\text{(disagreement points)}}   
	\end{align*}
	Inspired by the distance definition of Kim and Kopparty between a fraction
	word and a regular word (see \eqref{eq:KKdistance}), we define a
	similar modified distance between the fractional word $(g,\vw) \colon T \to
	\F_{<2}[u] \times (\Z_{\geq 0} \times \Z_{\geq 0})$ and and a word $h \colon T \to \F_{<2}[u]$ as
	follows. 
	
	\begin{align}
		\Gamma\left((g,\vw),h\right) := &\sum_{a \in
				A_0(g,h)}\left(2n - (d-\ell) - w(a,0)\right) + \sum_{a \in A_1(g,h)}
		\max\left\{\left( n- (d-\ell) - w(a,1) \right), w(a,0)\right\}\notag\\
		& + \sum_{a \in A_2(g,h)}\max\left\{w(a,0), w(a,1)\right\} \;,\label{eq:ourdistance}
	\end{align}
	(See \cref{def: gamma} for the exact definition of this modified
	distance.) This definition might seem complicated, however it is chosen
	so that it satisfies the following
	two important properties.
	\begin{itemize}
		\item {[\cref{lem: gamma comparison to usual distance}]} The fractional word $(g,\vw)$ and the ``correct'' polynomial
		$P_\ell$ satisfy
		\begin{align}\Gamma((g,\vw),P_\ell) &\leq \Delta_{\mathsf{mult}}^{(2)}(f,\enc^{(2)}(P)) < \frac12
			n^2\left(2-\frac{d}{n}\right)\;.\label{eq:prop1}
		\end{align}
		\item {[\cref{lem: gamma triangle}]} For any two distinct polynomials $Q_\ell$ and $R_\ell$ of
		degree $\ell$, we have
		\begin{align}
			\Gamma((g,\vw), Q_\ell) +
			\Gamma((g,\vw),R_\ell) & \geq n^2\left(2
			-\frac{d}{n}\right)\;. \label{eq:prop2}
		\end{align}
	\end{itemize}
	These two properties imply that the pair $(g,\vw)$ uniquely determines
	the polynomial $P_\ell$, which is in fact an alternative statement of the multiplicity SZ lemma.
	
	As in the case $s=1$, we are now left with the problem of efficiently extracting
	the polynomial $P_\ell$ from the fractional word
	$(g,\vw)$. In the next section, we describe how this was done
	by Kim and Kopparty and our modification of the same.
	
	%We remark that coming up with an appropriate definition of the generalized distance $\Gamma$ above is an important te
	\subsection{Extracting \texorpdfstring{$P_\ell$}{P l} from the fractional word}
	
	Let us restate the problem of $P_\ell$ extraction: we are given a
	fractional word $(g,w)\colon T \to \F \times \Z_{\geq 0}$ that
	satisfies \eqref{eq:KKprop1}, namely $\Delta((g,w), P_\ell) \leq \frac12
	n^2\left(1-\frac{d}{n}\right)$ (note that by \eqref{eq:KKprop2}, this
	uniquely defines $P_\ell$) and we have to find the polynomial
	$P_\ell$. Kim and Kopparty designed the weighted Reed-Solomon decoder,
	inspired by Forney's generalized minimum distance (GMD)  decoder for
	this purpose. To see this connection to Forney's GMD decoding, we
	first rewrite \eqref{eq:KKprop1} as follows:
	
	\begin{align*}
		\frac12 \cdot (n-(d-\ell)) \cdot \left( \sum_{a \in A_1} \ow(a) +
		\sum_{a \in A_0} (2-\ow(a))\right) \leq \frac12 n^2\left(1 -
		\frac{d}{n}\right)\;,
	\end{align*}
	where $\ow(a) := w(a)/\left(\frac12 \cdot (n-(d-\ell))\right) \in
	[0,1]$. Or equivalently,
	\begin{align*}
		\sum_{a \in A_1} \ow(a) +
		\sum_{a \in A_0} (2-\ow(a)) \leq \frac{n^2\left(1 -
			\frac{d}{n}\right)}{ (n-(d-\ell))} \leq n \left( 1 -
		\frac{\ell}{n}\right)\; .
	\end{align*}
	This can be further rewritten as
	\begin{align}
		\sum_{a \in T} \ow(a) + \sum_{a \in A_0} 2(1-\ow(a)) \leq n \left( 1 -
		\frac{\ell}{n}\right)\; .\label{eq:KKforney}
	\end{align}  
	The above inequality is very reminiscent of Forney's analysis of GMD
	decoding (c.f.,~\cite[Section~12.3]{GuruswamiRS}). Consider a uniformly random
	threshold $\theta \in [0,1]$ and erase all coordinates $a
	\in T$ such that $\ow(a) \geq \theta$. \eqref{eq:KKforney} can be now
	rewritten as
	\begin{align}
		\E_\theta\left[ \text{\# erasures}(\theta) + 2\cdot \text{\# errors}(\theta)\right]
		\leq n \left( 1 -
		\frac{\ell}{n}\right)\;. \label{eq:Forneyhyp}
	\end{align}
	This implies that there exists a threshold $\theta\in [0,1]$ that
	satisfies
	$\text{\# erasures}(\theta) + 2\cdot \text{\# errors}(\theta) \leq n
	\left( 1 - \frac{\ell}{n}\right)$. This is precisely the condition
	under which Reed-Solomon decoding (under erasures and errors) is
	feasible. Thus, the weighted Reed-Solomon decoder of Kim and Kopparty
	is as follows: consider every possible threshold $\theta \in [0,1]$
	and the corresponding erased word
	$g|_{\{a \in T \colon \ow(a) < \theta\}}$. Observe that there are at most
	$n$ distinct thresholds to consider. We can now run a standard
	Reed-Solomon decoder on this partially erased word to extract the
	polynomial $P_\ell$.
	
	This completes the description of the Kim-Kopparty decoder for Reed-Muller codes on grids.
	
	We would now like to adapt this weighted
	Reed-Solomon decoder to the multiplicity setting. We are now given
	the fractional word $(g,\vw) \colon T \to
	\F_{<2}[u] \times (\Z_{\geq 0} \times \Z_{\geq 0})$ with the promise
	of \eqref{eq:prop1}, namely $\Gamma((g,\vw),P_\ell) \leq
	\frac12 n^2(2-\frac{d}n)$ and we have to find the polynomial $P_\ell$. Recall that \eqref{eq:prop2} states that
	the promise uniquely determines $P_\ell$.
	
	A natural ``weighted
	univariate multiplicity code decoder'' would be the following:
	consider all possible step-thresholds $\thetabar =
	(\theta_0,\theta_1$) (or for general $s$, $(\theta_0,\theta_1,\ldots,\theta_{s-1})$) and erase coordinates of $g$
	accordingly to obtain a univariate multiplicity received word of
	varying multiplicities\footnote{We are intentionally vague at this
		point. For more details, see \cref{sec:wumd}.}. If this step-threshold
	$\thetabar$ is indeed correct (i.e., it retains a large fraction of
	agreements with $P_\ell$), we can show that we can adapt the classical
	Welch-Berlekamp decoder for Reed-Solomon codes to decode even in this
	setting. And we do precisely this in \cref{sec:BWmultiplicity}. We are
	still left with the problem of showing that there exists a step-threshold $\thetabar$ that works.
	
	\begin{figure}
		\centering
		\includegraphics{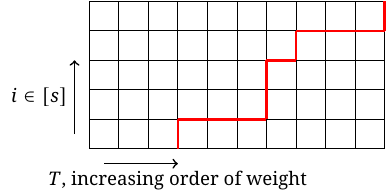}
		\caption{Step-threshold in weighted univariate multiplicity decoding}
		\label{fig:step threshold}
	\end{figure}
		
	How did Kim and Kopparty show that such a threshold exists in the
	$s=1$ setting? They showed by scaling the promise \eqref{eq:KKprop1}
	by $\frac12(n-(d-\ell))$ to obtain \eqref{eq:KKforney} which is
	equivalent to the statement \eqref{eq:Forneyhyp} that a random
	threshold $\theta$ works. Observe that the scaling
	$\frac12(n-(d-\ell))$ is precisely the unique-decoding radius of the
	$(d-\ell)$ degree Reed-Solomon code $G^{(a)}(y)$. In our setting
	($s>1$), there is not one code $G^{(a)}(y)$, but a whole family of
	them ($G^{(a,0)}(y)$ and $G^{(a,1)}(y)$ for the $s=2$ case). Recall
	that these are members of multiplicity codes of degree $d-\ell$ with
	multiplicity parameter 2 and 1 respectively. Their corresponding
	unique-decoding radii are $\frac12(2n-(d-\ell))$ and
	$\frac12(n-(d-\ell))$ (in terms of the multiplicity distance). In
	general, we might have $s$ different unique-decoding radii. It is
	unclear which of these scalings we should use.
	
	There is a more fundamental reason why one does not expect such a
	random-threshold argument to hold. Consider the following instance of
	the problem. Let $P$ be a degree $d$ polynomial for $2n/3< d <
	2n-2$. Let $B \subseteq T$ be a set of points of size
	$\left\lceil \frac{n}{2}\left(1-\frac{d}{2n}\right)\right\rceil
	-1$. Consider a word $f\colon T^2 \to \F_{<2}[u,v]$ formed by
	corrupting $\enc^{(2)}(P)$ on $T \times B$. The fraction of errors is
	$|B|/n < \frac12\left(1-\frac{d}{2n}\right)$ and hence we must be able
	to uniquely decode $P$ from $f$. Let us see how our suggested
	algorithm proceeds when it decodes along each column $x=a$. For each
	$a \in T$, the function $f_\ell^{(0)}(a,y)$ is corrupted on less than
	$\frac12\left(1-\frac{d}{2n}\right)$ and hence the univariate
	multiplicity code decoder for multiplicity parameter 2 correctly
	decodes this column. Furthermore, $w(a,0) = |B|$. However, while
	decoding $f_\ell^{(1)}(a,y)$ we observe that the number of errors is
	$|B|$ while the distance of the underlying code (multiplicity code
	with multiplicity parameter 1) is only $n\left(1- \frac{d}{n}\right)$
	which is less than $|B|$ for our choice of parameters (since $d > 2n/3$). So we could have
	corrupted enough points so that for each $a\in T$, $f^{(1)}_\ell(a,y)$
	is a valid codeword and yet different from the correct codeword in
	$\enc^{(2)}(P)$. In this case, all the $G^{(a,1)}(y)$ polynomials are
	erroneous, however $w(a,1) = 0$. We thus have for this setting
	$w(a,1) = 0$ and $w(a,0) = |B|$ for all $a \in T$. Consider the
	step-threshold $(|B|+1, 0)$ (by this we mean that we use the threshold
	$|B| + 1$ for $f_{\ell}^{(0)}$ and $0$ for $f_{\ell}^{(1)}$). Then the $f_\ell^{(1)}$ level is completely erased, while the $f_\ell^{(0)}$ level is completely un-erased. This object would be a received word for univariate multiplicity code decoding with varying multiplicities, in general; in this case since the first-derivative level is completely erased, it is an instance of Reed-Solomon decoding. Then, in fact, it is correctly decoded to obtain $P_\ell$, since each column $x=a$ was correctly decoded at the $f_\ell^{(0)}$ level, the only available level. Yet it is easy to check that the expected number of
	erasures plus two times the number of errors is too large for any
	reasonable choice of a random threshold. 
	
	How do we then show a step-threshold $\thetabar$ works without  
	resorting to the random-threshold argument?
	
	For starters, we could ask if we could give an alternative analysis of
	Forney's GMD decoding without appealing to the random-threshold
	argument. We show that this is indeed true. More formally, we show
	that if every threshold fails for Forney's GMD decoding, then the
	promise for GMD decoding is violated (i.e., the received word is more than half the distance of the concatenated code away from a codeword). We prove this as follows. Assume
	every threshold fails. We then construct a matching between agreement
	points and disagreement points (that is, correctly and incorrectly decoded blocks, following the inner decoding) such that every matched
	agreement-disagreement pair has the property that the weight in the
	disagreeing coordinate is at least that in the agreeing
	coordinate. Such a matching then immediately implies that the promise
	for GMD decoding is violated. See \cref{sec:forney} for further
	details. This alternative analysis of Forney's GMD decoding is not
	needed for our algorithm. It only serves as an inspiration for our
	``weighted univariate multiplicity code decoder'' just as the original
	analysis inspired the ``weighted Reed-Solomon decoder'' of Kim and
	Kopparty.
	
	Equipped with this alternative analysis of Forney's GMD decoding, we
	show a similar result for our setting. For contradiction, let us
	assume that every step-threshold $\thetabar$ fails to decode the
	polynomial $P_\ell$. This lets us construct a matching between
	``correct'' and ``incorrect'' coordinate-multiplicity pairs. The
	construction of this matching is considerably more involved than the
	construction of the corresponding matching in Forney's GMD decoding as
	we need to respect the monotonicity of the multiplicity
	coordinates. The matching is constructed via a delicate analysis using
	a generalization of Hall's Theorem (see \cref{thm:gen_hall}).  As in
	the Forney analysis, the matched agreement-disagreement pairs have the
	property that the multiplicity-distance of the disagreeing member is
	at least that of the agreeing member.  Once we have constructed such a
	matching, it is not hard to show that the promise \eqref{eq:prop1} is
	violated. Hence, at least one of the step-thresholds work and our
	``weighted univariate multiplicity code decoder'' works as
	suggested. This part of the analysis happens to be the most
	technically-challenging part of the paper and indeed this ``weighted univariate
	multiplicity code decoder'' forms the work-horse of our bivariate
	decoder (for more details, see \cref{sec:wumd}). The factor of $n^s$
	in the running time comes because this step runs over all possible
	step-thresholds and there are $O(n^s)$ of them. We conjecture that
	this step can be improved to $\poly(n,s)$.

	\subsection{The multivariate setting for \texorpdfstring{$m>2$}{m > 2}}
	
	The above presentation gave an outline of both the Kim-Kopparty
	decoder as well as ours in the bivariate setting. Kim and Kopparty
	extended the bivariate Reed-Muller decoder to larger $m$ by designing
	a weighted version of the Reed-Muller decoder which they then used
	inductively as suggested by the proof of the classical SZ Lemma. We
	however do not have a similar ``weighted multivariate multiplicity
	code decoder''. In particular, we do not even have a weighted version of
	the bivariate decoder. We get around this issue by performing a
	slightly different proof of the SZ Lemma from the textbook proof (see \cref{sec:msz}). This
	alternative proof of the SZ Lemma proceeds by viewing the
	polynomial as an $(m-1)$-variate polynomial with the coefficients
	coming from a univariate polynomial ring $\F[x_m]$ instead of as a
	univariate polynomial in $x_m$ with the coefficients coming from the
	$(m-1)$-variate polynomial ring $\F[x_1, \ldots, x_{m-1}]$. This
	slight modification allows us to work with just a ``weighted
	univariate multiplicity code decoder''. We note that a similar
	construction could have been obtained in the Reed-Muller setting too. We discuss this multivariate generalization in more detail in \cref{sec: multivariate algorithm}.

	\section{Multiplicity codes and main result} \label{sec: mult prelims}
	In this section, we give a formal statement of our main result. We start by building up the requisite notation and preliminaries including the notion of multiplicity codes and their properties which provide the key underlying motivation for the results in this paper.  
	
	\subsection{Notation}
	
	\begin{itemize}
		\item We use boldface letters like $\vx, \vy, \vz$ for $m$-tuples of
		variables. 
		\item $\F$ denotes the underlying field. 
		\item For a positive integer $\ell$, $[\ell] = \{0, \ldots, \ell -
		1\}$ (note $\ell \notin [\ell]$). 
		\item For $\ell \in \N$, $\F_{\leq \ell}[\vx]$ (and $\F_{<\ell}[\vx]$)
		denotes the set of
		polynomials in $\vx$, with coefficients in $\F$ of degree at most
		(strictly less than)
		$\ell$. 
		\item For $\ve \in \Z_{\geq 0}^m$, $\vx^{\ve}$ denotes the monomial
		$\prod_{j\in [m]} x_j^{e_j}$ and $|\ve|_1 := \sum e_i$.
	\end{itemize}

	\subsection{Multiplicity code}\label{sec:mult}
	We now state the definition of multiplicity codes. 
	\begin{definition}[Multiplicity code]\label{def:multiplicity code}
		Let $s, m \in \N$, $d \in \Z_{\geq 0}$, $\F$ be a field and $T \subseteq \F$ be a non-empty finite subset of $\F$. The $m$-variate order-$s$ multiplicity code for degree-$d$ polynomials over $\F$ on the grid $T^m$ is defined as follows.
		
		Let
		$E := \{\ve \in \Z_{\geq 0}^m \mid 0 \leq |\ve |_1 <
		s\}$.  Note that
		$|E| = \binom{s+ m -1}{m}$, which is the number of distinct
		monomials in $\F_{<s}[\vz]$. Hence, we can identify $\F^E$ with
		$\F_{<s}[\vz]$.
		
		The code is defined over the alphabet $\F^E$ and has block length $|T|^m$ with the coordinates being indexed by elements of the grid $T^m$. The code is an $\F$-linear map from  the space of polynomials $\F_{\leq d}[\vx]$ to $\left(\F_{< s}[\vz]\right)^{|T|^m}$, where for any $\va \in T^m$, the $\va^{th}$ coordinate of the codeword denoted by $\enc_{T^m}^{(s)}(P)|_{\va}$ is given by
		\[
		\enc_{T^m}^{(s)}(P)|_{\va} := P(\va + \vz) \mod \langle \vz \rangle^s. 
		\] 
		Thus, a codeword of this code can be naturally viewed as a function from $T^m$ to $\F_{<s}[\vz]$. 
	\end{definition}

	\begin{remark}\label{rmk: notation enc}
		For the ease of notation, we sometimes drop one or both of $s$ and
		$T^m$ from $\encst(P)$ when they are clear from the context. Also, we
		will use $n$ to refer to the size of $T$, i.e., $n= |T|$.
	\end{remark}
	
	Univariate multiplicity codes were first studied by Rosenbloom \&
	Tsfasman~\cite{RosenbloomT1997} and Nielsen
	\cite{Nielsen2001}. Multiplicity codes for general $m$ and $s$ were
	introduced by Kopparty, Saraf and Yekhanin~\cite{KoppartySY2014} where 
	they constructed locally decodable codes with high rate and small
	query complexity. In the context of local decoding, these codes are
	typically studied with the set $T$ being the entire field $\F$ (or a subfield of $\F$). However,
	in this work, we encourage the reader to think of $T$ being an
	arbitrary subset of $\F$.
	
	\begin{remark}We note that for every $\va \in \F^m$, $\left( P(\va + \vz) \mod
		\langle \vz \rangle^s \right)$ is a polynomial in $\F[\vz]$ of
		degree at most $s-1$, and therefore can be viewed as a function from the
		set $E$ to $\F$ given by its coefficient vector. Thus, the alphabet of
		the code is indeed $\F^E$. Moreover, the coefficients of $\left( P(\va
		+ \vz) \mod \langle \vz \rangle^s \right)$ are precisely the
		evaluation of the Hasse derivatives of order at most $(s-1)$ of $P$ at
		the input $\va$. Recall that for a polynomial $P \in \mathbb{F}[\vx]$,
		the $\ve^{th}$ Hasse derivative of $P$, denoted by $\frac{\hpartial
			P}{\hpartial \vx^{\ve}}$ is the coefficient of $\vz^{\ve}$ in
		$P(\vx+\vz)$. Typically, multiplicity codes are defined directly via
		Hasse derivatives, however, in this paper, it is notationally more
		convenient to work with \cref{def:multiplicity code}. Observe that
		with this notation, it is natural to think of a received word (an
		input to the decoding algorithm for multiplicity codes) as being given
		as a function $f\colon T^m \to \F_{<s}[\vz]$. 
	\end{remark}
	
	The distance of multiplicity codes is guaranteed by a higher order
	generalization of the Schwartz-Zippel Lemma, that was proved by Dvir,
	Kopparty, Saraf and Sudan~\cite{DvirKSS2013}. We need the following
	notation for the statement of the lemma.
	\begin{definition}[multiplicity at a point] For an $m$-variate
		polynomial $P \in \F[\vx]$ and a point $\va \in \F^m$, the
		multiplicity of $P$ at $\va$, denoted by $\mult(P, \va)$ is the
		largest integer $\ell$ such that $P(\va + \vz) = 0 \mod \langle \vz
		\rangle^{\ell}$ , or equivalently, the Hasse derivative
		$\frac{\hpartial P}{\hpartial \vx^{\ve}}(\va)$ is zero for all
		monomials $\vx^{\ve}$ of degree at most $\ell -1$.
	\end{definition}
	
	We now state the multiplicity Schwartz-Zippel lemma. 
	
	\begin{lemma}[multiplicity Schwartz-Zippel lemma~{\cite[Lemma 8]{DvirKSS2013}}]
		Let $P \in \F[\vx]$ be a nonzero $m$-variate polynomial of total degree at most $d$ and let $T \subseteq \F$. Then, \[ \sum_{\va \in T^m} \mult(P, \va) \leq d  \abs{T}^{m-1}.\] 
	\end{lemma}
	We note that the lemma immediately implies that the multiplicity codes
	as defined in \cref{def:multiplicity code} have distance at least $n^m
	(1-\frac{d}{ns})$.
	
	\subsection{A fine-grained notion of distance for multiplicity codewords}
	The multiplicity Schwartz-Zippel lemma states that for two distinct polynomials $P$ and $Q$ of total degree $d$, 
	\[ \sum_{\va \in T^m} (s-\mult(P-Q, \va)) > n^m \left( s- \frac{d}{n} \right).\]
	Thus, $\sum_{\va \in T^m} {(s-\mult(P-Q, \va))}$ is a candidate
	measure of distance between the encodings of $P$ and $Q$. However, it may be the case that
	$\mult(P-Q, \va)$ exceeds $s$, making the quantity negative. To get around this, we ``cap'' 
	$\mult(P-Q, \va)$ at $s$. That is, for each $\va \in T^m$, we take $s - \min\{\mult(P-Q, \va), s\}$.
	
	We certainly have 
	\[ \sum_{\va \in T^m}(s - \min\left\{\mult(P-Q, \va), s\right\}) \geq  \sum_{\va \in T^m} (s-\mult(P-Q, \va)) > n^m \left( s- \frac{d}{n} \right) \]
	
	We then naturally extend this measure of distance to a notion of
	distance between functions $f, g \colon T^{m} \to \F_{<s}[\vz]$ that might not necessarily be valid codewords of a multiplicity code. 
	%We start by defining this notion.  
	\begin{definition}\label{def:SZ distance}
		Let $T\subseteq \F$ be any set. For functions $f, g \colon T^m \to \F_{<s}[\vz]$, we define $\Delta_{\mathsf{mult}}^{(s)} (f, g)$  as
		
		\[ \Delta_{\mathsf{mult}}^{(s)} (f, g) := \sum_{\va \in T^m} (s - d^{(s)}_{\min}(f(\va) - g(\va)))\, ,\]
		where 
		for a polynomial $R$, $d^{(s)}_{\min}(R)$ is defined
		to be the minimum of $s$ and the degree of the minimum degree monomial
		with a non-zero coefficient in $R$. Note that if $R$ is identically
		zero, then $d^{(s)}(R)$ is set to $s$.
	\end{definition}
	
	As indicated above, the quantity $d^{(s)}_{\min}$ is related to the notion of multiplicity in the following sense: 
	\[
	d^{(s)}_{\min}(\enc_{T^m}^{(s)}(P)(\va) -\enc_{T^m}^{(s)}(Q)(\va)) = \min\{\mult(P-Q, \va), s\}\, .
	\]
	For brevity, we abuse notation slightly and sometimes drop one or more of the parameters $r, s, T^m$ from $\Delm^{(s)}(\enc_{T^m}^{(r)}(P),\enc_{T^m}^{(r)}(Q))$ when they are clear from the context. Note that the parameter $r$ in the encoding and $s$ in the $\Delm$ in $\Delm^{(s)}(\enc_{T^m}^{(r)}(P),\enc_{T^m}^{(r)}(Q))$ might be different from each other. 
	
	\begin{remark}\label{rmk:Delm for non-constant s}
		In \cref{def:SZ distance}, $s$ does not depend on $\va$ and is
		constant throughout. However, in the course of our analysis we will
		also encounter a scenario (see \cref{alg:gen mult decoder}) where $m =
		1$ and the multiplicity parameter is not constant and is given by a
		function $\sbar\colon T \to \Z_{\geq 0}$. The \cref{def:SZ distance} naturally extends to this case as follows.
		\[ \Delta_{\mathsf{mult}, T}^{(\sbar)} (f, g) = \sum_{\va \in T} (\sbar(\va) - d^{(s)}_{\min}(f(\va) - g(\va)))\,.\qedhere\]
	\end{remark}

	The following observation relates $\Delm^{(s)}$ with the standard definition of Hamming distance for multiplicity codes of order $s$. 
	\begin{observation}\label{obs:delm vs del}
		For any two polynomials $P, Q \in \F[\vx]$ of degree at most $d$, 
		\[
		\Delm^{(s)}(\encst(P), \encst(Q)) \leq s\cdot \Del(\encst(P), \encst(Q)) \, .
		\]
	\end{observation}
	Intuitively, if $\encst(P)$ differs from $\encst(Q)$ at a point 
	$\va \in T^m$, then the standard Hamming distance $\Del(\encst(P), \encst(Q))$ counts this point with weight one, whereas the new notion of distance $\Delm^{(s)}(\encst(P), \encst(Q))$ also takes into account 
	the lowest degree monomial in $\vz^{\ve}$ such that the coefficients of $\vz^{\ve}$ in $\encst(P)$ and $\encst(Q)$ are not equal to each other. This fine-grained structure will be crucially used in the analysis of our algorithm. 
	
	\subsection{Main Result}
	
	With these definitions in place, our main technical result can be stated as follows. 
	\begin{theorem}\label{thm:main tech result}
		Let $s, d, m \in \N$ and $T \subseteq \F$ be such that $d <
		s|T|$. Then, there is a deterministic algorithm (\cref{alg:outer multivariate})
		that runs in time $(sn)^{3m + s + O(1)}\cdot \binom{m+s-1}{s}$ and on input
		$f\colon T^m \to \mathbb{F}_{<s}[\vz]$ outputs the unique polynomial
		$P \in \mathbb{F}_{\leq d}[\vx]$ (if such a polynomial exists) such
		that
		\[\Delm^{(s)}(f, \encst(P)) < \frac{n^m}{2}\left(s-\frac{d}{n}\right),\]
		where $n = \abs{T}$.
	\end{theorem}
	
	Since this is stated in terms of the fine-grained distance, this is actually a strengthening of the main result as stated in the introduction.
	
	\section{Decoding univariate multiplicity codes with varying
		multiplicities}\label{sec:BWmultiplicity}
	In this section, we discuss an algorithm for decoding univariate
	multiplicity codes up to half their minimum distance based on the
	standard Welch-Berlekamp decoder\footnote{By the
		Welch-Berlekamp decoder we refer to the description of the Welch-Berlekamp algorithm~\cite{WelchB1986} provided by
		Gemmell and Sudan~\cite{GemmellS1992}.} for Reed-Solomon codes. While such
	decoders for univariate multiplicity codes are well known, the decoder
	discussed here handles a slightly more general scenario than
	off-the-shelf decoders of this kind, namely that the multiplicity
	parameter at each evaluation point is not necessarily the same. This
	slight generalization is necessary for our applications in this
	paper. %However, as we observe in this section, a natural adaptation of the Welch-Berlekamp decoding algorithm handles this scenario.
	
	\subsection{Description of the generalized univariate multiplicity decoder}
	
	We start with a description of the algorithm.

 	\begin{algorithm}%[H]
		\caption{Generalized Univariate Multiplicity Decoder}
		\label{alg:gen mult decoder}
\nonl 	\KwIn{$T \subseteq \F$ \Comment*[f]{set of evaluation points}}
            \myinput{$e$ \Comment*[f]{degree of the univariate polynomial}}
            \myinput{$\sbar \colon T \to \Z_{\geq 0}$ \Comment*[f]{number of
                    multiplicities}}
              \myinput{$h \colon T \to \F[z]$ such that for all $a \in T$, $h_a = \sum_{i\in[\sbar\left(a\right)]} h_a^{\left(i\right)} z^i$ \Comment*[f]{received word}}

		\KwOut{$R \in \F_{\leq e}[x]$ such that $\Delm^{(\sbar)}(h,\enc^{(\sbar)}(R)) <
			\frac{1}{2} (\sum_a \sbar(a) - e)$, if such an $R$ exists and $\bot$ otherwise}
		
%		\nonl \hrulefill \\
		Set $N \gets \sum_{a \in T} \sbar(a)$\,; \\
		Set $D \gets \lfloor \frac{1}{2} (N+e) \rfloor +1$ \,; \\
		Find a non-zero polynomial $Q(x, y) = B_0(x) + y\cdot B_1(x)$ such
		that \label{line:interpolation}\\
		\nonl\Indp $\bullet$ $\deg (B_0) < D$,\\
		\nonl $\bullet$ $\deg (B_1) < D-e$, and  \\
		\nonl $\bullet$ $\forall a \in T$, $Q(a+z, h_a) \equiv 0 \mod
		z^{\sbar(a)}$ \,;\\
		\Indm
		
		\leIf{the following three conditions are satisfied\\
			\nonl\Indp $\bullet$ $-B_0/B_1$ is a polynomial,\\
			\nonl $\bullet$ $-B_0/B_1$ has degree $\leq e$, and\\
			\nonl $\bullet$ $\Delm^{(\sbar)}(h,\enc^{(\sbar)}(-B_0/B_1)) <
			\frac{1}{2} (\sum_a \sbar(a) - e)$\\}
		{\KwRet{$-B_0/B_1$}}{\KwRet $\bot$\,.}
		
	\end{algorithm}
	We remark while the overall structure of the algorithm is essentially that of the Welch-Berlekamp based algorithms, the main point of difference is that the number of linear constraints imposed in the interpolation step (\cref{line:interpolation}) at any $a \in T$ is $\sbar(a)$, and might be different for  different $a \in T$. 
	
	\subsection{Correctness and running time of \texorpdfstring{\cref{alg:gen mult decoder}}{Algorithm 1}}
	We first show that the interpolation step (\cref{line:interpolation}) is possible, that is, a  non-zero polynomial $Q(x, y)$ satisfying the constraints exists. Then, we argue that if there is a polynomial $R$ of degree at most $e$ such that $\Delm^{(\sbar)}(h,\enc^{(\sbar)}(R)) < \frac{1}{2} (\sum_a \sbar(a) - e)$, then $R$ is indeed output by the algorithm. Moreover, from the check in step $4$, it is clear that the algorithm never outputs a polynomial that is \emph{far} from the received word. Thus, together, these claims imply the correctness of the algorithm. 
	
	We also note that the algorithm is efficient and runs in polynomial time in its input size, since all it needs is to solve a linear system of not-too-large size and a call to an off-the-shelf polynomial division algorithm, both of which can be done efficiently. We now proceed with the proof of correctness.

	\begin{claim}\label{clm: interpolation}
		For $D = \lfloor \frac{1}{2}(N+e) \rfloor + 1$, there exists a non-zero polynomial $Q(x,y) = B_0(x)+y\cdot B_1(x)$ with $\deg(B_0) < D$, $\deg(B_1) < D-e$ and $\forall a \in T$, $Q(a+z, h_a) \equiv 0 \mod z^{\sbar(a)}$, where $N = \sum_a \sbar(a)$. 
		%\[\frac{\hpartial^i (B_0)}{\hpartial(x^i)} (a) + \sum_{j=0}^i h_a^{(j)} \frac{\hpartial^{i-j} (B_1)}{\hpartial(x^{i-j}) }(a) = 0\]
		%for every $a \in T$ and $0 \leq i \leq \sbar(a)-1$, where $N = \sum_a \sbar(a)$. 
		Moreover, for any such non-zero solution, $B_1(x)$ is non-zero.
	\end{claim}
	
	\begin{proof}
		As is standard with decoding algorithms for various algebraic codes that are based on the polynomial method, we think of the constraints as a system of homogeneous linear equations in the coefficients of the unknown polynomials $B_0$ and $B_1$. 
		
		The number of variables in the linear system equals the number of
		coefficients we need to find across $B_0$ and $B_1$, which is
		$D+(D-e) = 2D-e$. For every $a \in T$, the constraint
		\[ Q(a+z, h_a) \equiv 0 \mod z^{\sbar(a)} \] is really $\sbar(a)$ many
		homogeneous linear constraints on the coefficients of $Q$, with there
		being one constraint corresponding to the coefficient of $z^i$ in
		$Q(a+z, h_a)$ being $0$ for every $i \in [\sbar(a)]$. Thus, the total
		number of homogeneous linear constraints is
		$\sum_{a\in T} \sbar(a) = N$. Hence, setting $2D-e >N$, e.g.,
		$D > \frac{1}{2}(N+e)$ ensures the existence of a non-zero solution.
		
		For the \emph{moreover} part, observe that if $B_1$ is identically
		zero, then the system of homogeneous linear constraints imposed on $Q$
		imply that $B_0$ vanishes with multiplicity at least $\sbar(a)$ for
		every $a \in T$. Thus,
		$\sum_{a \in T} \mult(B_0, a) = \sum_{a \in T} \sbar(a) = N > D \geq
		\deg(B_0)$. But this implies that $B_0$ must be identically zero, and
		hence $Q$ must be identically zero, which contradicts the non-zeroness
		of the solution.
	\end{proof}
	
	We now observe that any non-zero solution of the linear system, as is
	guaranteed by \autoref{clm: interpolation} becomes identically zero when
	$y$ is substituted by the \emph{correct} message polynomial $R$.
	
	\begin{claim}\label{clm: close enough poly staisfy eqn}
		If $R \in \F_{\leq e}[x]$ is such that $\Delm^{(\sbar)}(h,\enc^{(\sbar)}(R)) < \frac{1}{2}(N - e)$ and $Q$ is a non-zero polynomial satisfying the set of constraints in the algorithm, then $Q(x, R) \equiv 0$.
	\end{claim}
	
	\begin{proof}
		The linear constraints imposed in the interpolation step imply that for every $a \in T$, \[ Q(a + z, h_a) \equiv 0 \mod z^{\sbar(a)}. \]
		Now, if $(h_a - \enc^{(\sbar)}(R)(a)) = 0 \mod z^{u_a}$, i.e. $\enc^{(\sbar)}(R)$ and $h$ agree with multiplicity at least $u_a$ at $a \in T$, then 
		\[ 
		Q(a + z, R(a + z)) \equiv 0 \mod z^{\min\{u_a, \sbar(a)\}}. 
		\]
		
		Now, from the definition of $\Delm$ and from the hypothesis of the claim, we know that
		\[
		\Delm^{(\sbar)}(h,\enc^{(\sbar)}(R)) = \sum_{a \in T} (\sbar(a) - \min\{u_a, \sbar(a))\} < \frac{1}{2}(N-e)\, .
		\]
		
		If $Q(x, R(x))$ is a non-zero polynomial, then it is a non-zero polynomial of degree strictly less than $D$, by construction of $Q$. We will now show the sum of multiplicities of zeroes at all points is at least $D$ which implies $Q(x, R(x))$ is identically zero.
		
		\[
		\begin{split}
			\sum_{a \in T} \mult(Q(x, R(x)), a) & \geq \sum_{a \in T} \min\{u_a, \sbar(a)\} \\
			& > \sum_{a \in T} s(a) - \frac{1}{2} (N-e) \\
			& = N - \frac{1}{2} (N-e) \\
			& = \frac{1}{2}(N+e)
		\end{split}
		\]
		
		Note that $\sum_{a \in T} \mult(Q(x, R(x)), a)$ is an integer. Then we must have 
		\[ \sum_{a \in T} \mult(Q(x, R(x)), a) > \lfloor \frac{1}{2}(N+e) \rfloor +1 = D \qedhere \]
	\end{proof}
	
	We are now ready to complete the proof of correctness of the algorithm. %We start by restating \cref{lem: correctness generalized univariate mult decoding}. 
	\begin{theorem}\label{thm:correctness of generalized univariate multiplicity code decoder}
		If $R \in \F_{\leq e}[x]$ is such that $\Delm^{(\sbar)}(h,\enc^{(\sbar)}(R)) < \frac{1}{2}(N - e)$, 
		%and  $Q$ is a non-zero polynomial satisfying the set of constraints in Line $4$ of Algorithm, 
		then $R$ is correctly output by the algorithm and  if there is no such
		close enough $R$, then the algorithm outputs $\bot$. Moreover, the algorithm runs in time $N^{O(1)}$. 
	\end{theorem}
	\begin{proof}
		From \autoref{clm: interpolation}, we know that the linear system solver in the algorithm (\cref{line:interpolation})  always finds a non-zero polynomial $B_0(x) + yB_1(x)$ satisfying the linear constraints imposed in the algorithm, regardless of the existence of a codeword that is close enough to the received word. Moreover, in any such non-zero solution, $B_1$ is a non-zero polynomial. 
		
		We also know from \autoref{clm: close enough poly staisfy eqn} that any $R \in \F_{\leq e}[x]$ such that $\Delm^{(\sbar)}(h,\enc^{(\sbar)}(R)) < 1/2(N -e)$ satisfies
		\[
		B_0(x) + B_1(x)\cdot R(x) = 0 \, ,
		\]
		or, in other words $R = -B_0/B_1$, and is correctly output by the algorithm. 
		Since the algorithm performs a sanity check in the last line to make sure that $-B_0/B_1$ is indeed a low-degree polynomial such that the corresponding codeword is close to $h$ and outputs the computed solution only if this check is passed; else it outputs $\bot$. Thus, it never outputs an incorrect solution. 
		
		\cref{alg:gen mult decoder} really only needs to set up and solve a linear system of size $O(N)$ and perform some basic univariate polynomial arithmetic involving polynomials of degree at most $D < N$. Thus, it runs in time $N^{O(1)}$. 
	\end{proof}

	\section{Forney's generalized minimum distance decoding}\label{sec:forney}
	
	In this section, we give an alternative analysis of Forney's algorithm
	for decoding a concatenated code from half its minimum distance,
	assuming that there is an efficient algorithm for decoding the outer
	code from errors (scaled by a factor of $2$) and erasures up to its
	minimum distance, and that the inner code has an efficient maximum
	likelihood decoder. An example of such a setting is when the outer
	code is the Reed-Solomon code and the inner code has block length
	logarithmic in the total block length, and is over an alphabet of
	constant size. In this case, the outer code can be efficiently decoded
	from errors and erasures using, for instance, the Welch-Berlekamp
	algorithm, as long as
	\[
	2(\#\text{Errors}) + (\#\text{Erasures}) <  (\text{Minimum distance}) \, .
	\]
	For the inner code, one can just do a brute-force iteration over all codewords, and find the closest one. We recommend the reader to think of this example going forward even though the discussion here applies to a general concatenated code. 
	
	\subsection{Concatenated codes}
	We start with a definition of concatenated codes. 
	\begin{definition}[concatenated code]\label{def: concatenated codes}
		Let $q \geq 2, k\geq 1$ be natural numbers and let $Q = q^k$. Let
		$C_{out}\colon [Q]^K \to [Q]^N$ be a code of minimum distance $D$
		and let $C_{in}\colon [q]^k \to [q]^n$ be a code of minimum distance
		$d$. The concatenated code
		$C = C_{out}\circ C_{in} \colon [Q]^K \to [q]^{Nn}$ is defined as
		follows: Given a message $m \in [Q]^K$, we first apply $C_{out}$ to
		$m$ to get a codeword $m' \in [Q]^N$. Since $Q =q^k$, we identify
		each symbol of $m'$ with a vector of length $k$ over $[q]$, or
		equivalently, an element in the message space of $C_{in}$. We now
		encode each coordinate of $m'$ using $C_{in}$ to get a vector in
		$[q]^{Nn}$. See \cref{fig:concatenated encoding}.
	\end{definition}
	
	\begin{figure}[htp]
		\centering
		\includegraphics{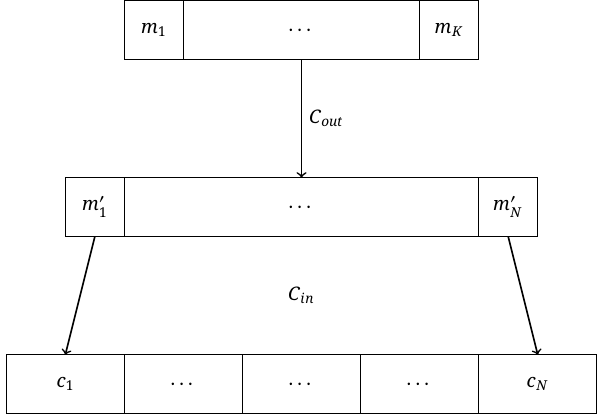}	
		\caption{Encoding of a concatenated code}
		\label{fig:concatenated encoding}
	\end{figure}
		
	As defined above, the concatenated code $C$ has minimum distance at least $Dd$. Forney designed an algorithm~\cite{Forney1965, Forney1966} to uniquely decode $C$ when the number of erroneous coordinates is less than $Dd/2$. Next, we briefly describe Forney's algorithm before discussing a slightly different analysis for it. 
	
	Let $f$ be the received word (obtained by making less than $Dd/2$
	errors), and let $c$ be the (unique) codeword of the concatenated code $C$ such that $\Del(f, c) < Dd/2$ where $\Delta$ refers to the usual Hamming distance.

	\begin{figure}[htp]
		\centering
		\includegraphics{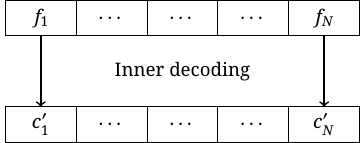}		
		\caption{Inner decoding of a concatenated code}
		\label{fig:concatenated inner decoding}
	\end{figure}	
	
	Notice that the received word $f$ consists of $N$ blocks of length $n$
	each over $[q]$, which we denote by $f_1, \dotsc, f_N$. Similarly, let
	$c_1, c_2, \ldots, c_N$ be the corresponding blocks in the nearest
	codeword $c$. The decoding algorithm starts by taking each of the
	blocks $f_i$ and using the maximum likelihood decoder for $C_{in}$ to
	find the codeword $c_i'$ of $C_{in}$ that is closest to $f_i$.
	This process is shown in \cref{fig:concatenated inner decoding}.  For some of the
	coordinates $i \in [N]$, $c_i'$ equals $c_i$ and for others
	$c_i' \neq c_i$. For each $i \in [N]$, let $m'_i$ be such that
	$C_{in}(m'_i) = c'_i$. The natural next step would be to use a decoder for
	the outer code $C_{out}$ with the input $m' = (m'_1,m_2',\ldots,m_N')$
	and hope to show that the output equals $c$ as is intended. When the
	number of errors in $f$ (i.e. $\Delta(f, c)$) is less than $Dd/4$ then
	this algorithm can indeed be shown to work and output the correct
	message $m$ corresponding to the codeword $c$.
	
	In~\cite{Forney1965, Forney1966}, Forney built upon this simple and
	natural decoder to design an efficient algorithm with error tolerance
	improved from $Dd/4$ to $Dd/2$. The key idea was to assign a weight
	\[
	w(i)\coloneqq \min\left\{\frac{\Delta(f_i, c_i')}{d/2}, 1\right\}
	\]
	to each block $i$; the intuition being that the weight $w(i)$ is an
	indicator of the number of errors in the block $i$. Thus, the higher
	the weight for a coordinate $i$, the lower is the confidence in $c_i'$
	being $c_i$.
	
	Forney showed that if $\Delta(f, c) < Dd/2$, then there is a threshold
	$\theta \in \Z_{\geq 0}$ such that the vector
	$m'' = (m''_1,m''_2,\ldots,m''_N) \in ([Q]\cup \{\bot\})^{N}$ where
	$m''_i$ equals $m_i'$ if $w(i) \leq \theta$ and is an erasure ($\bot$)
	otherwise has the property that when compared to the correct codeword
	$c$, the sum of the number of erased blocks and twice the number or
	erroneous blocks is less than $D$. Thus, given an error-erasure
	decoder for $C_{out}$ from half the minimum distance, we can recover
	the correct message $m$.\footnote{We remark that just showing the
		existence of a good threshold $\theta$ is sufficient, since without
		loss of generality, $\theta$ can be taken to be equal to one of the
		weights in the set $\{w(i) \colon i \in [N]\}$. Thus we can
		efficiently try all these $N$ values and check if the resulting
		message gives a codeword close enough to the received word $f$.}
	
	The key technical task in the correctness of above algorithm is to
	prove the existence of such a threshold $\theta$ given that
	$\Delta(f, c) < Dd/2$. Forney proves this using a convexity argument
	which shows that if all the thresholds in
	$\{ w(i) \colon i \in [N] \}$ fail, then the number of errors in $f$
	must be at least $Dd/2$. This convexity argument has a randomized
	interpretation, which in turn shows that a random threshold succeeds
	(see the manuscript by Guruswami, Rudra and Sudan
	\cite[Section~12.3]{GuruswamiRS} for this randomized interpretation).
	
	In the next section we give an alternative proof of existence of a \emph{good} threshold, thereby proving the correctness of Forney's algorithm. This alternative proof serves as the inspiration for our proof of \cref{thm:main tech result}, and as far as we understand appears to be different from the above-mentioned proofs.
	
	\subsection{An alternative analysis of Forney's GMD decoding}\label{subsec: alternate_Forney}
	We now give an alternative proof of the existence of a good threshold.

	\begin{theorem}[Forney]\label{THM:GOOD_THRESHOLD_FORNEY}
		Let $f, c, c', w, D, d$ be as defined above. If $\Delta(f,c) < Dd/2$, then there exists a $\theta \in [0,1]$ such that 
		\[
		2\left| \left\{ i \in [N] \colon w(i) \leq  \theta, c_i' \neq c_i
		\right\} \right| + \left| \left\{ i \in [N] \colon w(i) > \theta \right\} \right| < D \, .
		\]
	\end{theorem}
	We recall that the condition in the conclusion of the theorem is sufficient for the outer decoding to be done by our assumption on the outer code $C_{out}$. For instance, if the outer code is the Reed-Solomon code, then the Welch-Berlekamp decoder can be used for this. 
	
	Notice that all weights $w(i)$ lie between $0$ and $1$ and, as mentioned earlier, can be thought of as representing the uncertainty in that block. 
	Let $A$ and $B$ be defined as follows:
	\[A := \setdef{i \in [N]}{c_i = c_i'}, \qquad B := \setdef{i \in
		[N]}{c_i \neq c_i'}.\]
	That is, $A$ is the set of blocks where the inner decoding is correct, and $B$ is those where the decoding is incorrect. Further, for a threshold $\theta$, we define $A_{\theta} = A \cap \setdef{i \in [N]}{w(i) \leq \theta}$. $B_{\theta}$ is defined similarly. Notice that the number of erased blocks for a given $\theta$ 
	is $N - |A_{\theta}| - |B_{\theta}|$.
	
	We say that a threshold $\theta$ is bad if
	\[
	2\left| \left\{ i \in [N] \colon  w(i) \leq \theta, c_i' \neq c_i \right\} \right| + \left| \left\{ i \in [N] \colon  w(i) > \theta \right\} \right| \geq D \, .
	\]
	Or, in the new notation, $2| B_{\theta}| + (N - (|A_\theta| + |B_{\theta}|)) \geq D$, or equivalently,  
	\[
	\abs{B_{\theta}} \geq \abs{A_{\theta}} - (N-D) \, .
	\]
	% Note that if the above inequalities do not hold for a threshold $\theta$, then we can indeed use the error-erasure decoding of the outer code to correctly find $c$. 
	We note the above in the following observation.
	
	\begin{observation} \label{obs:undecodable}
		For a fixed $\theta \in [0,1]$, if the threshold $\theta$ is bad, then $\abs{B_{\theta}} \geq \abs{A_{\theta}} - (N-D)$.
	\end{observation}
	
	\begin{proof}[Proof of Theorem~\ref{THM:GOOD_THRESHOLD_FORNEY}]
		Suppose for contradiction that all 
		thresholds $\theta \in [0,1]$ are bad. By \autoref{obs:undecodable}, we have $\abs{B_{\theta}} \geq \abs{A_{\theta}} - (N-D)$ for every $\theta \in [0,1]$.
		
		First, we observe that the size of $A$, i.e.,  the number of
		correctly decoded blocks, must be more than $(N-D)$. Otherwise, since
		$\abs{A}+\abs{B} = N$, the number of errors $\abs{B}$ is at least
		$D$. Further, for every $i \in B$, the $i^{th}$ block in $f$ must
		have had at least $d/2$ errors for the maximum likelihood decoder for
		$C_{in}$ to have output a $c_i' \neq c_i$. But this means that
		$\Delta(f, c)$ is at least $Dd/2$ which is a
		contradiction. Therefore, we can write $\abs{A} = N-D+u$ for some
		positive integer $u$.
		
		Now, write the indices in $A$ in increasing order of their weights, i.e., according to $w$, to get a sequence $i_1, i_2, \ldots, i_{N-D+u}$. We do the same with $B$ to get a sequence $j_1, j_2, \ldots, j_{D-u}$. If many blocks have the same weight, we place them in some arbitrary order.

		We first claim that  $u\leq D-u$ and for each $k \in \{1, 2, \ldots, u\}$, $w(j_k) \leq w(i_{N-D+k})$.
		In other words, if the indices $i \in \{1,2,\ldots, N\}$ are written left to
		right in increasing order of weights $w(i)$, then for every $k \in \{1, 2, \ldots, u\}$, the index $j_k$ is to the left of the index $i_{N-D+k}$ as indicated in the picture below.
		
		\begin{figure}%[htp]
			\centering
			\includegraphics{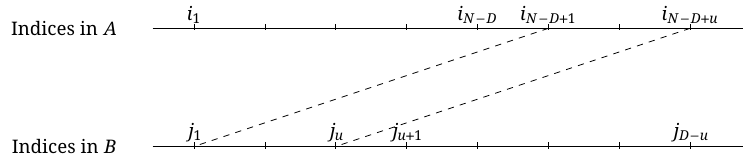}					
			\caption{Pairing of indices in the concatenated code}
			\label{fig:forney pairing}
		\end{figure}
		
		To see this, let us assume that there is a
		$k \in \{1, 2, \ldots, u\}$, such that either $w(j_k) > w(i_{N-D+k})$
		or $j_k$ does not exist. Consider a threshold $\theta$ which picks up
		everything including and to the left of $i_{N-D+k}$. Then, we have
		$\abs{A_{\theta}} \geq N-D+k$ and $\abs{B_{\theta}} \leq
		k-1$. This
		contradicts our assumption that
		$\abs{B_{\theta}} \geq \abs{A_{\theta}} - (N-D)$ for every $\theta$.
		
		Using the above claim, we can ``pair'' each of the last $u$ indices in $A$ with the first $u$ of $B$ (ordered by weight). When we pair $i_{N-D+k}$ with $j_k$, we have $w(j_k) \leq w(i_{N-D+k})$. For the incorrectly decoded index $j_k$, the distance of the block $f_{j_k}$ from $c_{j_k}$ is at least $\frac{d}{2}(2-w(j_k))$. This follows from the triangle inequality of the Hamming distance: $\Delta(f_{j_k},c_{j_k})+\Delta(f_{j_k},c'_{j_k})\geq \Delta(c_{j_k},c'_{j_k})\geq d$.
		
		For a correct index $i_{N-D +k}$, the distance between $f$ and $c$ on
		this block is by definition $\frac{d}{2}w(i_{N-D+k})$. Thus, the total distance between $f$ and $c$ on the pair $(j_k, i_{N-D+k})$ of blocks together is at least 
		\[
		\frac{d}{2}(2-w(j_k)) + \frac{d}{2}w(i_{N-D+k}) \geq d,
		\]
		where the inequality follows from the fact that $w(j_k) \leq w(i_{N-D+k})$. We now have $u$ pairs contributing a distance of at least $d$ each. In addition, there are still $D-2u$ incorrectly decoded blocks with indices in $B$, each contributing at least $d/2$ to $\Delta(f,c)$. Thus, we have
		\[ 
		\Del(f,c) \geq du + (D-2u)\frac{d}{2}  = \frac{Dd}{2} \, ,
		\]
		which is a contradiction.
	\end{proof}
	
	As mentioned earlier, Forney's original proof~\cite{Forney1965,
		Forney1966} uses a convexity argument, while more recent
	presentations of this argument reinterprets this as a randomized
	argument (erasing the block $i$ with probability $w(i)$ for each $i$)
	to show that there is a good threshold $\theta$ (see
	\cite[Section~12.3]{GuruswamiRS}). For the analysis of our bivariate decoder,
	we adapt the \emph{pairing} argument used in the proof of
	\cref{THM:GOOD_THRESHOLD_FORNEY} to analyse the weighted univariate
	multiplicity code decoder in \cref{sec:wumd}. It is unclear to us if the previous
	convexity-based proofs of  \cref{THM:GOOD_THRESHOLD_FORNEY} can be
	adapted for our application.
	
	\section{Weighted univariate multiplicity code decoder}\label{sec:wumd}
	
	In this section, we describe an algorithm (\cref{alg:weighted multiplicity code decoder}) for decoding weighted
	univariate multiplicity codes, i.e. given a received word $g\colon T
	\to \F_{<r}[z]$, and a weight function $w\colon  T\times [r]\to
	\Z_{\geq 0}$ (indicating the uncertainty in each coordinate), this algorithm finds a low-degree polynomial $R$ such
	that the encoding of $R$ as a univariate multiplicity code of order
	$r$ evaluated on $T$ is \emph{close enough} to the received word
	$h$. 
	
	The notion of \emph{close enough} here is not defined in
	terms of either the Hamming distance or the multiplicity distance, but a weighted notion of distance parameterized by the weight function $w$. This notion of distance, denoted by $\Gamma$, was previously introduced in \cref{eq:ourdistance}.
	It serves as a proxy for distance in our analysis, turns out to be
	crucial in the overall analysis of \cref{alg:weighted multiplicity code decoder} and will be discussed in detail in \cref{sec:weighted distance}. 
	
	More precisely, this decoder gets as input a function $g\colon T \to
	\F_{<r}[z]$ and a set of weights $w\colon T \times [r] \to \Z_{\geq
		0}$. Recall that we referred to this pair $(g,w)$ as a ``fractional word'' in
	the algorithm overview. It is also given as input a degree parameter
	$\ell$ and multiplicity parameter $r$ besides the global degree
	parameter $d$, the global dimension $m$ and global multiplicity parameter $s$. We have two
	sets of degree and multiplicity parameters (the local and global) as
	the intermediate algorithms will be running decoders on degree and
	multiplicity parameters different from the global ones, but they
	do need to know the global parameters. The weighted decoder
	then returns a polynomial $R \in \F_{\leq \ell}[x]$ such that
	$\Gamma_w^{d, \ell, s}(g, R) < \frac{n^2}{2}\left(s - \frac{d}{n}\right)$.
	%     if such an $R$ exists and the zero polynomial otherwise. %Here,
	%     $\Gamma$ refers to the modified distance mentioned in the
	%     algorithm overview and is defined formally in \cref{def:
		%       gamma}. This decoder is invoked in \cref{line:wumd} of the
	%     bivariate decoding algorithm in
	% \cref{alg:outer} and discussed in more detail in
	% \cref{sec:wumd}.
	
	\cref{alg:weighted multiplicity code decoder} will be used as a subroutine by the final decoder (i.e., \cref{alg:outer} in the bivariate $m=2$ case or \cref{alg:outer multivariate} in the general multivariate case). Ideally, \cref{alg:weighted multiplicity code decoder} should be oblivious of whether it is being used by the bivariate decoder or the multivariate decoder. Unfortunately, this is not the case for our \cref{alg:weighted multiplicity code decoder} and we need to feed it as input several global parameters $m, s$, and $d$\footnote{We believe this dependence on global parameters should be removed, but do not know how to do so.}. In this section we will work with the global dimension $m$ being 2 for the sake of analysis, though the algorithm is stated in terms of general $m$. The general $m$ setting will be discussed in \cref{sec: multivariate algorithm}.
	
	% \cref{alg:weighted multiplicity code decoder} takes as input the degree $\ell$ and multiplicity $r$ corresponding to the current iteration, as well as the global degree $d$, global multiplicity $s$ and global dimension $m$. The received word is a function $g\colon T
	% \to \F_{<r}[z]$, and a weight function $w\colon  T\times [r]\to
	% \Z_{\geq 0}$ is also given, 
	
	\subsection{A notion of weighted distance and its properties}\label{sec:weighted distance}
	We start with the following definition of a distance measure between the received word $g$ and the encoding of a polynomial $R$, using the weight $w$. The case for $s=2$ was discussed in \cref{sec: overview} (see \cref{eq:ourdistance}).
	\begin{definition} \label{def: gamma}
		Let $d, \ell, s \in \N$ be parameters with $d \geq \ell$ and let $T
		\subseteq \F$ be a subset of size $n$. 
		Let  $r:=s-\lfloor \frac{d-\ell}{n} \rfloor$. Let $R \in \F[x]$ be a
		univariate polynomial of degree at most $\ell$, $g\colon T \to
		\F_{<r}[z]$ and $w\colon  T \times [r] \to \Z_{\geq 0}$ be functions
		such that for every $(a, i) \in T \times [r]$ we have $w(a,i) \leq
		\frac{n}{2}\cdot \left((s-i) - \frac{d-\ell}{n}\right)$. 
		
		Then, $\Gamma_w^{s, d, \ell}(g, R)$ is defined as follows. 
		\begin{multline*}
		\Gamma_w^{s, d, \ell}(g, R) :=  \left(\sum_{i = 0}^{r-1}\sum_{a \in A_i(g,R)} \max \left\{\left(n\left((s-i) - \frac{d-\ell}{n}\right) - w(a,i)\right), \max_{j < i} w(a, j) \right\}\right)\\
        + \sum_{a\in A_r(g,R)} \max_{j < r} w(a, j)
		\end{multline*}
		where for every $i \in [r+1]$
		\[
		A_i(g,R)=\setdef{a\in T}{\max{\setdef{j\in [r+1]}{g(a)=R(a+z) \mod \langle z\rangle^{j}}}=i}. \qedhere
		\] 
		
	\end{definition}
	Observe that for $i\in [r+1]$, the set $A_i$ collects those locations in $T$ where $g$ and $R$ agree up to derivatives of order exactly $i-1$ and no further. For $i=0$, $A_0$ is the set of locations where they disagree at the $0^{th}$ derivative itself, i.e. the evaluation level, and the $\max_{j < i} w(a, j)$ term can be thought of as $- \infty$.
	
	Now, we prove some properties of $\Gamma_w^{s, d, \ell}$. These properties will turn out to be useful in proving the correctness of \cref{alg:weighted multiplicity code decoder}.

	Below we prove a triangle-like inequality for $\Gamma$ which we will use to show uniqueness of the output of \cref{alg:weighted multiplicity code decoder}. This is the point alluded to in \cref{eq:prop2}.
	
	\begin{lemma}[Triangle-like inequality for $\Gamma$]\label{lem: gamma triangle}
		Let $d, \ell,r,s \in \N$ be parameters with $\ell \leq d < n$, $r =  s -
		\lfloor \frac{d-\ell}{n} \rfloor $, and let $T \subseteq \F$ be a
		subset of size $n$. Let $Q, R \in \F[x]$ be univariate polynomials of
		degree at most $\ell$, and $g\colon T \to \F_{<r}[z]$ and $w\colon T
		\times [r] \to \Z_{\geq 0}$ be functions such that for every $(a, i)
		\in T \times [r]$, $w(a,i) \leq \frac{n}{2}\cdot \left((s-i) - \frac{d-\ell}{n}\right)$. 
		
		If $Q \neq R$, then 
		\[
		\Gamma_w^{s,d,\ell}(g, Q) + \Gamma_w^{s,d,\ell}(g, R) \geq n^2\left(s - \frac{d}n\right) \, .
		\]
		
	\end{lemma}
	\begin{proof}
		For $i\in[r+1]$, let $A_i(Q,R)$ be the set of points $a \in T$ such that
		\[
		A_i(Q,R)=\setdef{a\in T}{\max{\setdef{j\in [r+1]}{Q(a+z)=R(a+z) \mod \langle z\rangle ^{j}}}=i}.
		\]
		
		Further, for each $i \in [r+1]$, let $\tau_i = \abs{A_i(Q, R)}$. Since the sets $A_i(Q, R)$ for each $i$ are all disjoint, we have $\sum_{i=0}^r \tau_i = \abs{T} = n$. 
		
		In addition, since $Q$ and $R$ are distinct polynomials of degree at most $\ell$, using the multiplicity SZ lemma we have:
		\[
		\Delm^{(r)} (\enc^{(r)}_T(Q), \enc^{(r)}_T(R)) \geq rn-\ell.
		\]
		Using the definition of $\Delm^{(r)}$, this can be re-written as 
		\[
		\sum_{i=0}^{r-1}  \tau_i \cdot (r-i) \geq rn-\ell \, .
		\]
		That is,
		\[ r \sum_{i=0}^{r} {\tau_i} - \sum_{i=0}^{r} i \cdot {\tau_i}  \geq rn-\ell \, .
		\]
		Rearranging and using $\sum_{i=0}^{r} \tau_i = n$, we obtain 
		%\[ r(n-\tau_r) \geq rn - \ell + \sum_{i=0}^{r-1} i \tau_i, \]
		
		%which yields 
		\[ 
		% \sum_{i=0}^{r-1} i \cdot \tau_i+ r \tau_r = 
		\sum_{i=0}^{r} i \cdot \tau_i \leq \ell \, .
		\] 
		
		Now, consider an $a \in A_i(Q, R)$ for some $i\in [r]$. (The case of $a \in A_r(Q, R)$ will be explained shortly.) Hence, $Q(a+z)\neq R(a+z)\mod \langle z\rangle^{i+1}$. Therefore, both $g(a)=Q(a+z) \mod \langle z \rangle^{i+1}$ and $g(a)= R(a+z)\mod \langle z\rangle^{i+1}$ can't simultaneously hold. In other words, there is a $j \leq i$ such that $a \in A_j(g,R)$ or $a\in A_j(g,Q)$.
		
		Without loss of generality, we assume $a \in A_j(g, Q)$ with $j \leq i$. In addition, $a \in A_k(g, R)$ for some $k$. Again we can assume $j \leq k$ (otherwise, swap the roles of $Q$ and $R$).
		
		As $a \in A_j(g, Q)$, we claim that the contribution of the term corresponding to $a$ to $\Gamma_w^{s,d,\ell}(g, Q)$, which we denote by $\Gamma_w^{s,d,\ell}(g, Q)_a$ is at least $\left((s-j)n - (d-\ell) - w(a, j)\right)$, because that is one of the terms in the maximum, in the definition of $\Gamma_w^{s,d,\ell}$.
		
		As for $\Gamma_w^{s,d,\ell}(g, R)_a$, since $a$ is in $A_k(g, Q)$ with $k \geq j$, two cases can happen. If $j \neq k$, then the contribution of $a$ is at least $w(a, j)$, since that is one of the terms in the maximum. Else, if $k = j$, then the contribution is at least $(s-j)n - (d-\ell) - w(a, j)$, which, by the condition on $w(a,j)$ in \cref{def: gamma} is at least $w(a,j)$.
		
		Then, the sum of the distances at $a$ is at least 
		\[ \begin{split}
			\Gamma_w^{s,d,\ell}(g, Q)_a + \Gamma_w^{s,d,\ell}(g, R)_a & \geq (s-j)n - (d-\ell) - w(a, j) + w(a, j) \\
			& = (s-j)n - (d-\ell) \\ 
			& \geq (s-i)n - (d-\ell) \, .
		\end{split} \]
		since $j \leq i$.
		
		If $a \in A_r(Q, R)$, we cannot get an expression of this form, since its contribution to $\Gamma$ will be $\max_{j < r} w(a, j)$. Since the contribution from this set is nonnegative, and we are trying to get a lower bound, it is sufficient to consider the contributions from $A_i(Q, R)$'s for $i < r$. 
		
		By the above arguments we arrive at the following.
		\[  
		\begin{split}
			\Gamma_w^{s,d,\ell}(g, Q) + \Gamma_w^{s,d,\ell}(g, R) & \geq \sum_{i=0}^{r-1} \tau_i ((s-i)n - (d-\ell))  \\
			& = \sum_{i=0}^{r-1} \tau_i (sn - (d-\ell)) -n \sum_{i=0}^{r-1} i \tau_i   \\
			& \geq (sn-(d-\ell))(n - \tau_r) -  n(\ell  - r \tau_r) \\
			& = n(sn-d) + \tau_r (rn + (d-\ell) - sn) \\
			& \geq n^2(s-d/n) \, .
		\end{split}
		\]
		using the relations $ \sum_{i=0}^r \tau_i = n$ and $ \sum_{i=0}^r i \tau_i \leq \ell$, and the definition of $r$: since $r \geq s - \frac{d-\ell}{n}$, the coefficient of $\tau_r$ in the penultimate expression is non-negative.
	\end{proof}

	\subsection{Weighted Univariate Multiplicity Code Decoder}
	
	For the Weighted Univariate Multiplicity Code Decoder we are given a word $g\colon T \to \F_{<r}[z]$ and a weight function $w\colon T\times [r]\to \Z_{\geq 0}$ along with parameters $d,\ell$ and $s$ with $\ell \leq d$ and $r = s - \lfloor \frac{d-\ell}{n} \rfloor $. Further, for every $(a, i) \in T \times [r]$: $w(a,i) \leq \frac{n}{2}\cdot\left((s-i) - \frac{d-\ell}n\right)$.
	The algorithm returns a polynomial $R$ of degree at most $\ell$ such
	that $\Gamma_w^{d, \ell, s}(g, R) $ is smaller than
	$\frac{n^2}{2}(s-\frac{d}{n})$. Notice that by \cref{lem: gamma
		triangle} there can only be one such $R$. 
	
	Recall that higher the value of $w(a,i)$, lower is our confidence on the $(i-1)^{th}$ derivative specified by $g$ at $a$. 
	For every $a\in T$, we set  $ \omega(a) \gets \max_{i \in [r]} w(a, i)$, i.e., the maximum distrust in any derivative specified by $g$ at $a$.
	For every step-threshold $\thetabar = (\theta_0, \theta_1, \ldots,
	\theta_{r-1}) \in [sn/2]^r$ with $\theta_0 \geq \theta_1 \geq \dotsb
	\geq \theta_{r-1}$ and for every $a\in T$  we retain $g(a)$ up to
	degree $i$ (equivalently derivatives up to order $i$) such that
	$\omega(a)\leq \theta_i$. Let the retained set of $(a,i)$'s be
	$U_{\thetabar}$; we then call \cref{alg:gen mult decoder} on $g$
	restricted to $U_{\thetabar}$. Then, we check whether the polynomial
	returned by this  step satisfies $\Gamma_w^{d, \ell, s}(g, R) < \frac{n^2}{2}(s-\frac{d}{n})$ and output it if it does.\\
	
	\setcounter{AlgoLine}{0}
	\begin{algorithm}[H] 
		\caption{Weighted Univariate Multiplicity Code Decoder}\label{alg:weighted multiplicity code decoder}
		\nonl\KwIn{$T \subseteq \F, |T| = n$ \Comment*[r]{set of evaluation points}}
        \myinput{$d, s, m$ \Comment*[f]{global degree, multiplicity and dimension resp.}}
        \myinput{$\ell, r$ with $\ell\leq d$ and $r =  s - \lfloor \frac{d-\ell}{n} \rfloor $  \Comment*[f]{actual degree and multiplicity resp.}}
        \myinput{$g\colon T \to \F_{<r}[z]$ \Comment*[f]{received word}}
        \myinput{$w\colon T\times [r] \to \Z_{\geq 0}$  satisfying $w(a,i) \leq
                  \frac{n^{m-1}}2\left(s-i - \frac{d-\ell}n\right), \; \forall (a, i)$}
        \myinput{\Comment*[f]{weight
		 			function}}
		\KwOut{$R \in \F_{\leq \ell}[x]$ such that $\Gamma_w^{d, \ell, s}(g,
			R) < \frac{n^m}{2}(s-\frac{d}{n})$, if such an $R$ exists and $0$ otherwise.}
%		\nonl \hrulefill \\
		\For{$a \in T$}{
			Set $ \omega(a) \gets  \max_{i \in [r]} w(a, i)$\,;
		}
        \For{every step threshold $\thetabar = (\theta_0, \theta_1, \ldots, \theta_{r-1})
			% \in [sn/2]^r
			$ such that $\theta_0 \geq \dotsb \geq
			\theta_{r-1}$ }{
            \Comment*[r]{There are at most $\binom{n+r}{r}$ step
				thresholds and the algorithm goes over each one of them.}
			Set $U_{\thetabar}\gets \setdef{(a, i)\in T \times [r]}{\omega(a) \leq \theta_i \ , \ a \in T, \ i \in [r]}$\,; \\
			\For{$a \in T$}{
				Set $\sbar(a) \gets  \max \setdef{i \leq r-1}{\omega(a) \leq \theta_i}+1$\,;
			}
			Run Generalized Univariate Multiplicity Code Decoder (\cref{alg:gen mult decoder}) on $(T, \ell, \sbar, g \vert_{U_{\thetabar}})$ where $g\vert_{U_{\thetabar}} \colon T \to \F_{<r}[z]$ is defined as $g\vert_{U_{\thetabar}}(a) = g(a) \mod z^{s(a)}$ to obtain $R$ (if $R=\bot$ set $R \leftarrow 0$)\,; \\
			\leIf{$\Gamma_w^{d, \ell, s}(g, R) < \frac{n^m}{2}(s-\frac{d}{n})$}{
				\KwRet{$R$} \label{line:check_gamma_in_wumd}}
			{\KwRet $0$\,.}
		}
	\end{algorithm}
	
	\begin{remark}
		The number of step-thresholds is $O({\binom{n+s}{s}})$. This can be seen by considering $\theta_0, \theta_1 - \theta_0, \theta_2 - \theta_1, \dotsc, \theta_{r-1}-\theta_{r-2}, n - \theta_{r-1}$. Each of these $r+1$ quantities is a non-negative integer and their sum is $n$, so the number of solutions is ${\binom{n + r}{r}}$, and $r \leq s$. 
	\end{remark}
	
	\subsection{Proof of correctness of \texorpdfstring{\cref{alg:weighted multiplicity code decoder}}{Algorithm 2}}
	
	The analysis of this decoder is inspired by the alternative analysis of
	Forney's GMD decoding mentioned in the previous section. 
	
	Notice that the algorithm proceeds by trying every monotone threshold vector
	$\thetabar$. (We only consider monotone threshold vectors in the algorithm and following analysis.) Hence, it suffices to show that one $\thetabar$ exists
	which can be used to do the decoding. We first characterise when a
	threshold vector can be used in the call to \cref{alg:gen mult
		decoder}. We then prove that such a threshold vector exists  in \cref{LEM:THRESHOLD_EXISTS}. For this we need the following notation. 
	Suppose there is a polynomial $R \in \F_{\leq \ell}[x]$ such that
	$\Gamma_w^{d, \ell, s}(g, R) < \frac{n^2}{2}(s-\frac{d}{n})$:
	\cref{lem: gamma triangle} implies that there can be at most one such
	$R$. 
	
	Given this polynomial $R$, we partition the set $T \times [r+1] = A
	\uplus B$ as follows. Let $A$ be the set of $(a,i)$ such
	that $g(a)=R(a+z) \mod \ip{z}^{(i)}$: in other words, at the
	location $a$ all derivatives of $R$ till order $i-1$ match with $g$. And
	let $B = (T\times [r+1]) \setminus A$. Let ${\thetabar}$ be a vector of
	thresholds. 
	
	Observe that if $(a, i) \in A$, then $(a, j) \in A$ for all $j < i$. Then, since $B = (T\times [r+1]) \setminus A$, if $(b, j) \in B$, then $(b, i) \in B$ for all $i > j$. 
	
	Define $A_\thetabar := U_\thetabar \cap A$ and similarly
	$B_\thetabar := U_\thetabar \cap B = U_\thetabar \setminus A_\thetabar$.
	
	\begin{observation}\label{obs:decodable}
		If it holds that $ \abs{A_{\thetabar}} > \abs{B_{\thetabar}} + \ell $, then \cref{alg:gen mult decoder} can decode using ${\thetabar}$. We call such threshold vectors \emph{good}. 
	\end{observation}
	
	This is because, \cref{alg:gen mult decoder} will decode using $\thetabar$ whenever  $\Delm^{(\sbar)}(g|_{U_{\thetabar}},\enc^{(\sbar)}(R)) < \frac{1}{2} (|U_{\thetabar}| - \ell)$.
	% where $h$ is the restriction of $g$ to $U_{\thetabar}$. 
	With our notation, this means $\abs{B_{{\thetabar}}} < \frac{\abs{U_{\thetabar}}-\ell}{2}$, since $\Delm^{(\sbar)}(g|_{U_{\thetabar}},\enc^{(\sbar)}(R)) = \abs{B_{{\thetabar}}}$.
	Rearranging this and using $\abs{U_{\thetabar}} = \abs{A_{\thetabar}}+\abs{B_{\thetabar}}$ gives the above characterization.
	
	The following lemma, which we prove in the next section, shows that there is at least one good step-threshold $\thetabar$.
	
	\begin{lemma} \label{LEM:THRESHOLD_EXISTS}
		Let $g$ be the received word. If $R$ is such that $\Gamma_w^{d, \ell, s}(g, R) < \frac{n^2}{2}(s-\frac{d}{n})$, then there is a good vector of thresholds $\thetabar$ such that \cref{alg:weighted multiplicity code decoder} returns $R$ in the iteration corresponding to $\thetabar$.
	\end{lemma}
	
	Armed with the above lemma we are now ready to prove the correctness of \cref{alg:weighted multiplicity code decoder}.
	
	\begin{theorem} \label{thm: correctness of wumd}
		Let $g\colon T \to \F_{<r}[z]$ be a received word and $R$ a degree $\ell$ polynomial such that $\Gamma_w^{d, \ell, s}(g, R) < \frac{n^2}{2}(s-\frac{d}{n})$. Then, \cref{alg:weighted multiplicity code decoder} returns the polynomial $R$. Further, \cref{alg:weighted multiplicity code decoder} runs in time $(sn)^{s + O(1)}$ where $n$ is the size of the set of evaluation points, $T$. 
	\end{theorem}
	
	\begin{proof}
		The algorithm proceeds by trying every step-threshold $\thetabar$. By \cref{LEM:THRESHOLD_EXISTS}, there is a good vector of thresholds $\thetabar$ that can be used to find $R$ in the call to \cref{alg:gen mult decoder}. Hence, the algorithm finds $R$ within the given distance if one exists. Also, the algorithm never outputs an incorrect $R$ due to the check at \cref{line:check_gamma_in_wumd} and the fact that if one such $R$ exists then it is unique by \cref{lem: gamma triangle}.
		
		The running time of the Algorithm is determined by the $O(sn)^r$ iterations of the for-loop over all possible step-thresholds. By \cref{thm:correctness of generalized univariate multiplicity code decoder}  each such iteration requires $(nr)^{O(1)}$ time. As $r\leq s$, the overall running time is $(sn)^{s+O(1)}$.
	\end{proof}
	
	\subsection{Proof of \texorpdfstring{\cref{LEM:THRESHOLD_EXISTS}}{Lemma 6.5}}
	
	The proof of \cref{LEM:THRESHOLD_EXISTS} requires a few claims and
	definitions. We begin by showing that $|A|>\ell$.
	
	\begin{claim} \label{clm: size of A}
		Let $g$ be any received word and $R$ be a degree $\ell$ polynomial with $\Gamma_w^{d, \ell, s} (g, R) < \frac{n^2}{2}(s-\frac{d}{n})$.  Let $A$ be defined as above, that is, the set of locations $(a,i)$ such
		that $g(a)=R(a+z) \mod \ip{z}^{(i)}$.  Then, $\abs{A} > \ell$. 
	\end{claim}
	
	\begin{proof}
		Suppose for contradiction that $\abs{A} \leq \ell$. We will show that the error is more than promised.
		
		Write $A$ as a disjoint union, $A = \biguplus_{i = 0}^{r} A_{\geq i}(g, R)$ where $A_{\geq i}(g,R)=A\cap (T\times \{i\})$.
		
		Informally, $A_{\geq i}(g, R)$ is the set of locations where $g$ and $R$ agree up to derivatives of order $i-1$. Let the size of $A_{\geq i}(g, R)$ be $\eta_{i-1}$, so that we have $n = \eta_0$ elements which are (trivially) ``correct'' up to the ``$(-1)^{th}$'' order derivative, $\eta_1$ up to the $0^{th}$ order derivative (i.e. evaluation level), $\eta_2$ up to the first-order derivative and so on.
		
		Note that $\sum \eta_i \leq \ell$, by our assumption. 
		Additionally, $\eta_0 \geq \eta_1 \geq \dotsb \geq \eta_{r}$ because if $(a,i) \in A$ then for all $0 \leq j<i$, $(a,j)\in A$.
		Hence the number of coordinates where $g$ and $R$ disagree at the $0^{th}$ order derivative is $n-\eta_1$ and the total contribution to $\Gamma$ from such coordinates is at least $\frac{1}{2}(n-\eta_1)(ns-(d-\ell))$. This is because the contribution from each term is at least $\frac12(ns-(d-\ell))$. To observe this, note that in \cref{def: gamma}, in the first term in the maximum, the weight being subtracted is at most $\frac{1}{2}(ns-(d-\ell))$, so the maximum never dips below this quantity. Further, there are at least  $(n-\eta_2)-(n-\eta_1)=\eta_1-\eta_2$ coordinates where $g$ and $R$ agree at the $0^{th}$ order derivative but disagree at the first order derivative: such coordinates contribute to $\Gamma$ at least $\frac{1}{2}(\eta_1-\eta_2)(n(s-1)-(d-\ell))$. In this manner, we get the total distance to be 
		
		\[  
		\begin{split}
			\Gamma_w^{d, \ell, s} (g, R) & \geq \frac{1}{2}(\eta_0-\eta_1)(ns-(d-\ell)) + \frac{1}{2}(\eta_1-\eta_2)(n(s-1)-(d-\ell)) + \\
			&  \hspace{7cm} \dotsb + \frac{1}{2}(\eta_{r-1}-\eta_{r})(n(s-r+1)-(d-\ell))\\
			& = \frac{1}{2} \sum_{i=0}^{r-1} \delta_i ( n(s-i) - (d-\ell)) 
		\end{split}
		\]
		where $\delta_i = \eta_{i} - \eta_{1+1}$. Note that each $\delta_i \geq 0$ and $\sum_{i=0}^{r-1} \delta_i 
		\leq n$ (by a telescoping sum). In addition, $\sum_{i=0}^{r-1} i \delta_i =  (\sum_{i=2}^{r-2} \eta_i)-(r-1)\eta_{r}  \leq \ell$, by assumption. 
		
		The rest of this proof is identical to the argument in the proof of \cref{lem: gamma triangle}. 
		\[  
		\begin{split}
			\Gamma_w^{s,d,\ell}(g, R) & \geq \frac{1}{2} \sum_{i=0}^{r-1} \delta_i ((s-i)n - (d-\ell))  \\
			& = \frac{1}{2} \sum_{i=0}^{r-1} \delta_i (sn - (d-\ell)) -\frac{n}{2} \sum_{i=0}^{r-1} i \delta_i   \\
			& \geq \frac{1}{2} (sn-(d-\ell))n -  \frac{n}{2}\ell \\
			& \geq \frac{1}{2} n^2  \left(s- \frac{d}{n} \right) \, .
		\end{split}
		\]
		which contradicts the assumption.
	\end{proof}
	
	Now we know $\abs{A} > \ell$, so we can write $\abs{A} = \ell+u$ for some positive $u$.
	For $a \in T$, let $i_R(a)$ be the $i$ such that $a \in A_i(g, R)$. Recall that for every $i \in [r+1]$
	\[
	A_i(g,R)=\setdef{a\in T}{\max{\setdef{j\in [r+1]}{g(a)=R(a+z) \mod \langle z\rangle^{j}}}=i}.
	\]
	Informally, $g$ and $R$ agree up to derivatives of order $i-1$ at $a$, but disagree at the $i^{th}$ order derivative.
	We construct a pairing as in our reproof of Forney's result in \cref{subsec: alternate_Forney}; however,  we will need it to be a ``good'' pairing to help in the error analysis. We provide the definition below.
	
	\begin{definition}[good pairing] \label{def: good pairing}
		Let $a_0, \dotsc, a_{k-1}, b_0, \dotsc, b_{k-1} \in T$. We say the $k$ pairs given by $(a_0,
		b_0), \dotsc, (a_{k-1}, b_{k-1})$ is a good pairing if it satisfies
		the following conditions:

		\begin{enumerate}
			\item $\abs{\set{a_i\colon \ i \in [k]}\cup \set{b_i\colon \ i \in [k]}} = 2k$. That is, they are all distinct. 
			\item $\omega(b_j) \leq \omega(a_j)$ for all $j \in [k-1]$  where as before  $ \omega(a) = \max_{i \in [r]} w(a, i)$.
			This is analogous to the weight condition in the Forney case.
			\item $i_R(b_j) < i_R(a_j)$ for all $j \in [k]$. This is analogous to one being correct (from $A$) and the other being corrupted (from $B$) - here, we need one to be correct up to more derivatives than the other. See \cref{fig:good pairing}
			\item $\sum_{j=0}^{k-1} (i_R(a_j) - i_R(b_j)) \geq u$. This is
			analogous to the number of pairs being $u$ (that is,  $\abs{A} -
			\ell$). (While there are $2k$ locations being paired with each
			other, the quantity of interest is the total difference in the
			derivatives to which $a_i$'s and $b_i$'s are correct. This has
			to be at least $u$). \qedhere
		\end{enumerate}
	\end{definition}
	
	\begin{figure}%[htp]
		\centering
		\includegraphics{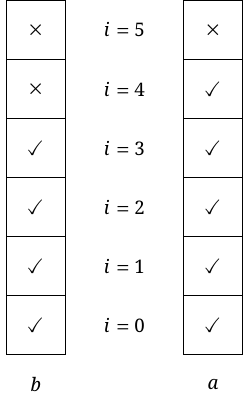}
		\caption{When $a$ and $b$ can be in a good pairing, with $\omega(b) \leq \omega(a)$}
		\label{fig:good pairing}
	\end{figure}
	
	\begin{lemma} \label{LEM:GOOD_PAIRING_EXISTS}
		If \cref{alg:weighted multiplicity code decoder} fails to find $R$ for every threshold vector $\thetabar$, then there exists a good pairing among the elements of $T$.
	\end{lemma}
	
	The proof of this lemma requires the following generalization of
	Hall's theorem (which can be found for instance in the textbook on
	Matching Theory by Lov{\'a}sz and Plummer). 
	
	\begin{theorem}[Generalisation of Hall's theorem~{\cite[Theorem~1.3.1]{LovaszP1986}}] \label{thm:gen_hall}
		In a bipartite graph $G = (\calL, \calR, E)$, if the maximum matching has size at most $m$, then there is a subset $U \subseteq \calL$ such that \[ \abs{U} - \abs{N(U)} \geq \abs{\calL} - m\;.\]
	\end{theorem}
	
	\begin{proof}[Proof of Lemma~\ref{LEM:GOOD_PAIRING_EXISTS}]
		
		Consider the following bipartite graph with left and right partite sets $\calL= A$ and $\calR =B$, respectively.
		It will be useful to stratify $\calL=\biguplus_{i = 0}^{r} \calL_i$ where  $\calL_i= A_{\geq i}(g, R)$;  recall that $A_{\geq i}(g,R)=A\cap (T\times \{i\})$. Similarly, we stratify $\calR=\biguplus_{i = 0}^{r}\calR_i$ where $\calR_i = (T \times \{i\}) \setminus \calL_i$. From earlier, $(a,i)\in \calL_i$ implies that for all $j<i$, $(a,j)\in \calL_j$, and inversely for $\calR$.
		
		The edge set of the graph is given as (See \cref{fig:edge_structure}):
		\[E = \setdef{((a, i),(b, i))}{(a, i) \in \calL_i \ , (b, i) \in \calR_i, \  \omega(a) \geq \omega(b) \ , \ 0 \leq i \leq r}.
		\] 
		
		It will be
		useful to visualise the vertices of $\calL$ and $\calR$, which are of
		the type $(a,i)$, ordered first according to their strata, i.e., $i$,
		and then within each stratum according to the $\omega(a)$ values. See \cref{fig:edge_structure}.
		
		We first show that this graph has a matching of size $u$, using the generalization of Hall's theorem from \cref{thm:gen_hall}. We will then convert this matching into a good
		pairing using \autoref{clm: matching to good pairing}.
		
		\begin{claim}
			The graph defined above has a matching of size $u$, where $\abs{\calL} = \ell + u$.
		\end{claim}
		
		\begin{subproof}
			
			\begin{figure}%[htp]
				\centering
				\includegraphics{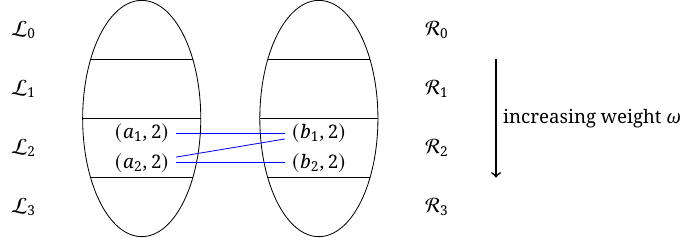}
				\caption{Edge structure of the bipartite graph in \cref{LEM:GOOD_PAIRING_EXISTS}.}
				\label{fig:edge_structure}
			\end{figure}
			
			Suppose for contradiction that the maximum matching in the graph has size $u-1$.
			Appealing to \cref{thm:gen_hall} with  $m=u-1$, we have a witness set $U \subseteq \calL$ such that $\abs{U} - \abs{N(U)} \geq \abs{\calL} - u + 1 = \ell +1 $.
			
			Let $U_i$ be the part of $U$ that lies in each $\calL_i$. For each $U_i$, observe that we may as well take the first (ordered by weight $\omega$ --- observe that this weight is independent of $i$) $\abs{U_i}$ elements in $\calL_i$, and this does not increase the neighbourhood. This is because $(a, i)$ of weight $\omega(a)$ is connected only to vertices $(b, i)$ with $\omega(b) \leq  \omega(a)$.
			Hence, we can assume each $U_i$ is a prefix of $\calL_i$ (i.e., the
			first $|U_i|$ elements of $\calL_i$). 
			
			Further, since $(a,i)\in \calL_i$ implies that for all $j<i$, $(a,j)\in \calL_j$ we can do the following transformation to $U$ without increasing the size of $N(U)$. If for $i>j$ we have $|U_i|>|U_j|$, then, we can replace $U_i$ and  $U_j$ with $U'_i = \set{(a,i)\colon (a,j) \in U_j }$ and $U'_j = \set{(a,j)\colon (a,i) \in U_i }$ respectively. This does not increase the neighbourhood. Indeed, consider some new neighbour $(b, j)$ caused by an edge $((a, j), (b, j))$. We must have $(a, j) \in U'_j \setminus U_j$. This newly included $(a, j)$ originated from some $(a, i)$ in $U_i \setminus U'_i$. Since $j < i$ and $(b, j) \in B$, we must also have $(b, i) \in B$. Then, the edge $((a, i), (b, i))$ also exists, and $(b, i)$ was previously in the neighbourhood of $U$ but no longer is. 
			
			We now have a structure on $U$ such that for every $i$, the first $\abs{U_i}$ elements of $A$, ordered by weight $\omega$, are in $U$. Further, for $i > j$, if $(a, i) \in U_i$ then $(a, j) \in U_j$. With this structure on $U$ we can now create a valid step-threshold $\thetabar$ so that $A_{\thetabar}$ contains $U$ (and exactly equals $U$ if the weights $\omega$ are all distinct) : indeed, set $\theta_i$ to be $\max \{ \omega(a) : (a, i) \in U_i \}$ (and if $U_i$ is empty, set $\theta_i$ to be $-\infty$). Further, $B_{\thetabar}=N(U)$. To see this, take $(b, i) \in \calR_i$ with $\omega(b) \leq \theta_i$. Then there would be an edge $((a, i), (b, i))$ for the $(a,i)$ with $\omega(a) = \theta_i$ and hence $(b, i) \in N(U)$. Conversely, if $(b, i) \in N(U)$, there is an edge $((a, i), (b, i))$ with $\omega(b) \leq \omega(a) \leq \theta_i$ and hence $(b, i) \in B_\thetabar$. (See \cref{fig:matching to threshold}). 
			However, as $\abs{A_\thetabar} - \abs{B_\thetabar} \geq |U|-|N(U)|\geq \ell +1$, this contradicts our assumption that there is no good step-threshold $\thetabar$. Hence, we must have a matching of size at least $u$ in the bipartite graph.			
		\end{subproof}
		
		To finish the proof of \cref{LEM:GOOD_PAIRING_EXISTS} we will need \autoref{clm: matching to good pairing} which essentially transforms this matching into a good pairing.
	\end{proof}
	
	\begin{figure}[htp]
		\centering
		\includegraphics{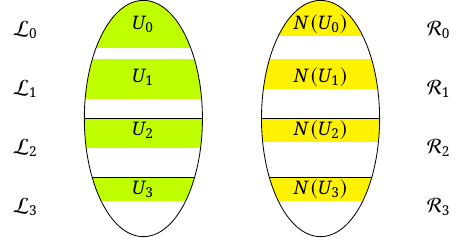}
		\caption{Extracting a threshold from the matching}
		\label{fig:matching to threshold}
	\end{figure}
	
	Next, we proceed to show how to extract a good pairing from a matching obtained in the bipartite graph defined above.
	
	\begin{claim} \label{clm: matching to good pairing}
		Given a maximum matching of size $u$ in the graph defined in \cref{LEM:GOOD_PAIRING_EXISTS}, a good pairing can be extracted from it. 
	\end{claim}
	
	\begin{proof}
		The given matching $M$ consists of $u$ edges of the form $((a, i), (b, i))$ with $\omega(b) \leq \omega(a)$. We want to extract a good pairing $(a_0, b_0), \dotsc, (a_{k-1}, b_{k-1})$.
		
		Recall the conditions from \cref{def: good pairing}. For Condition 1, we need to be able to read off the $k$ distinct pairs. We will add $(a, b)$ to the pairing if we have the edges $((a, j), (b, j))$ for every $i_R(b) < j \leq i_R(a)$ in $M$. This will also take care of Condition 4 since we are given exactly $u$ edges in the matching $M$. Conditions 2 and 3 are guaranteed by the edge structure.
		
		If $M$ consists of exactly the edges of this type, we are done. But, $M$ is an arbitrary matching and we have no guarantees on its structure. We will now see how to extract a pairing from it regardless, by making some modifications to $M$ while maintaining its size.
		
		We will now use a visualization as in \cref{fig:good pairing}, that is, every $a \in T$ has a tower of height $s$. If the matching $M$ contains an edge $(a, i), (b, i)$, we will think of the $a$ and $b$ towers being connected by an edge at level $i$.  
		
		We need, for every $a$ and $b$ in $T$, the matching $M$ to either match $(a, j)$ to $(b, j)$ for every $j$ between $i_R(b)$ and $i_R(a)$, or for none of them. (In the first case, $(a, b)$ is in the pairing; otherwise, it is not.)  The problem case is if they are matched on some levels but not on others --- in this case, we will call $a$ and $b$ a \emph{bad} pair of blocks. We will \emph{correct} this pair by either adding the missing edges or disconnecting $a$ and $b$ completely.
		
		Arrange the elements of $T$ in increasing order of the weight $\omega$. We will process every pair in co-lexicographic order, since for $(a, b)$ to be in a good pairing, we need $\omega(a) \geq \omega(b)$. Hence, when we are trying to correct a pair $(a, b)$, we can assume there are no bad pairs $(a', b')$ with $\omega(b') < \omega(b)$, or $\omega(b') = \omega(b)$ and $\omega(a') < \omega(a)$.
		
		Now, say we are at the stage of processing a bad pair $(a, b)$, and say they are matched at some other levels, but not the level $j$, that is, the edge $((a, j), (b, j))$ is not in the matching $M$. 
		
		$(a, j)$ could be unmatched in $M$, or it could be matched to some $(b', j)$. Notice that because of the order in which we are proceeding, we cannot have $\omega(b') < \omega(b)$. Otherwise, $(a, b')$ form a bad pair which we would have encountered earlier. Also, by definition, we cannot have $\omega(b') > \omega(a)$. Hence we must have $\omega(b) \leq \omega(b') \leq \omega(a)$.
		
		As for $(b, j)$, it could be unmatched in $M$, or matched to some other $(a', j)$. Again, if $\omega(a') < \omega(a)$, there is a bad pair $(a', b)$ which would have been encountered before $(a, b)$, so we can eliminate this case. So we must have $\omega(a) \leq \omega(a')$.
		
		We will now go through the possibilities one by one.
		
		\textit{Case 1:} $(a, j)$ and $(b, j)$ are both unmatched in $M$. (shown in \cref{fig:bad case 1})
		
		Then they can be matched with each other, which increases the size of the matching which is a contradiction since $M$ was a maximum matching.
		
		\begin{figure}[htp]
			\centering
			%%%%%%%%%%%%%CASE 1: both unmatched		
			\includegraphics{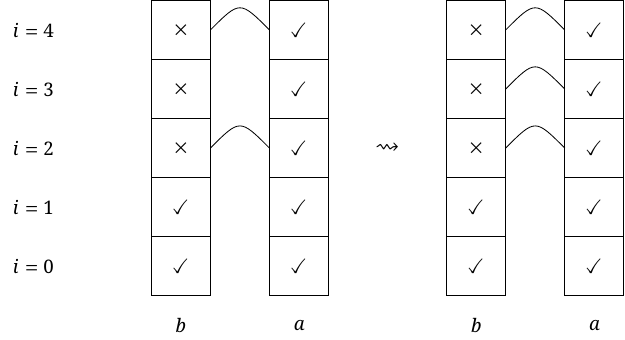}			
			\caption{Case 1: $(a, 3)$ and $(b, 3)$ both unmatched}
			\label{fig:bad case 1}
		\end{figure}
		
		\emph{Case 2:} $(b, j)$ is matched with some $(a', j)$ with $\omega(a') \geq \omega(a)$. (shown in \cref{fig:bad case 2})
		
		Then, we can remove that edge from $M$ and add the edge $((b, j), (a, j))$ to the matching. (If $(a, j)$ is already matched to $(b', j)$ in $M$, we remove it and add an edge between $(b', j)$ and $(a', j)$.) This may create a new bad pair of higher weight, which can be dealt with subsequently.
		
		\begin{figure}[htp]
			\centering
			%%%%%%%%%% CASE 2: b matched with c > a
			\includegraphics{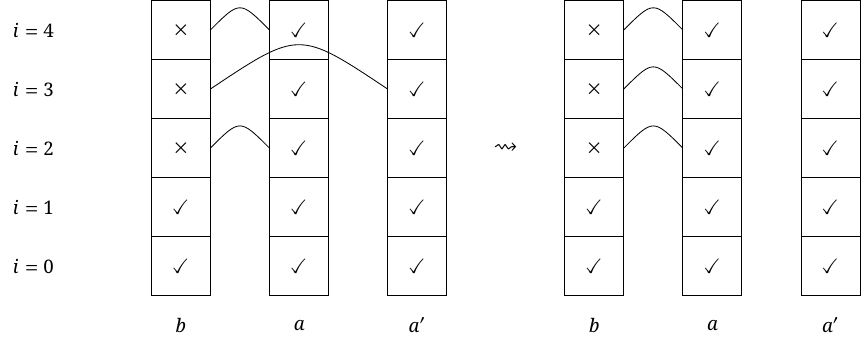}			
			\caption{Case 2: $(b, 3)$ is matched with $(c, 3)$ instead of $(a, 3)$}
			\label{fig:bad case 2}
		\end{figure}
		
		%\emph{Case 3:} $(b, j)$ is matched with some $(c, j)$ with $\omega(c) < \omega(a)$. (shown in \cref{fig:bad case 3})
		
		%In this case, $(c, b)$ is also a bad pair of blocks, which would have been processed first since $\omega(c) < \omega(a)$. Hence this case is not possible. 
		
		% \begin{figure}[htp]
			% \centering
			% \import{./figures/}{bad_case_3}
			% \caption{Case 3: $(b, 3)$ is matched with $(c, 3)$ instead of $(a, 3)$}
			% \label{fig:bad case 3}
			% \end{figure}
		
		\emph{Case 3:} $(b, j)$ is unmatched in $M$, and $(a, j)$ is matched to $(b', j)$. (shown in \cref{fig:bad case 3})
		
		\begin{figure}[htp]
			\centering
			%%%%%%%%%%%%%%% CASE 4: b unmatched, a matched to smth else
			\includegraphics{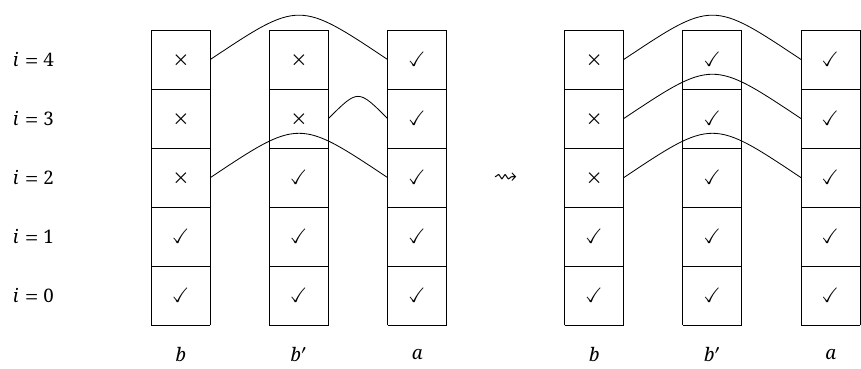}							
			\caption{Case 3: $(a, 3)$ is matched with $(b', 3)$ instead of $(b, 3)$}
			\label{fig:bad case 3}
		\end{figure}
		Then the edge between $(a, j)$ and $(b', j)$ can be removed and the edge between $(a,j)$ and $(b, j)$ added to $M$. Again, any new bad pairs created will be dealt with subsequently.
		
		This process will terminate since $T$ is finite. Then, $M$ will be of the form we saw above, and the good pairing can be read off from it.
	\end{proof}
	
	Now, we show that if there is a good pairing in the elements of $T$, then, $\Gamma_w^{d, \ell, s}(g, R) \geq \frac{n^2}{2}(s-\frac{d}{n})$. This is shown by transforming the received word $g$ and the weight function $w$ into a different word $g'$ and weight function $w'$, by changing the values of $g$ and $w$ at the locations where the good pairing exists. In doing so, we ensure $\Gamma(g,R)$ is non-increasing.
	We then reach a scenario where $|A|\leq \ell$ for this new word $g'$. Then, by \autoref{clm: size of A} we get that  $\Gamma_w^{d, \ell, s}(g, R) \geq \Gamma_{w'}^{d, \ell, s}(g', R)\geq \frac{n^2}{2}(s-\frac{d}{n})$.  
	
	\begin{claim} \label{clm: error decreases with replacement}
		Assume we have a good pairing of size $k$, $(a_0, b_0), \dotsc,
		(a_{k-1}, b_{k-1})$, in the received word $g$. Construct a new
		received word $g'$ as follows: for each $j$ in $[k]$ (recall that
		$i_R(a_j) > i_R(b_j)$) in the original received word, corrupt the
		values of $\set{g(a_j,i)\colon i\in \set{i_R(b_j)+1,\ldots,i_R(a_j)}
		}$ such that $i_R(a_j) = i_R(b_j)$.  Further, for all $i \in [r]$ let
		\[w'(a_j,i)=w'(b_j,i)=\min{\set{ \frac{n}2\cdot \left((s-i_R(b_j)) - \frac{d-\ell}n\right),\frac{n}2\cdot\left((s-i) - \frac{d-\ell}n\right) }}.\] 
		At all other locations which are not part of the pairing, $g'$ and $w'$ are exactly the same as $g$ and $w$.
		See \cref{fig:pairing transform}. We note $w'$ is a valid weight function.% (Also, any location not part of the pairing is left as it is.) 
		Then, the error as per our measure actually $\emph{does not increase}$, that is $\Gamma_{w'}^{d, \ell, s} (g', R) \leq \Gamma_{w}^{d, \ell, s} (g, R)$. Furthermore, $\Gamma_{w'}^{d, \ell, s}(g', R)\geq \frac{n^2}{2}(s-\frac{d}{n})$.  
	\end{claim}
	
	\begin{cfigure}%[htp]
		\centering
		\includegraphics{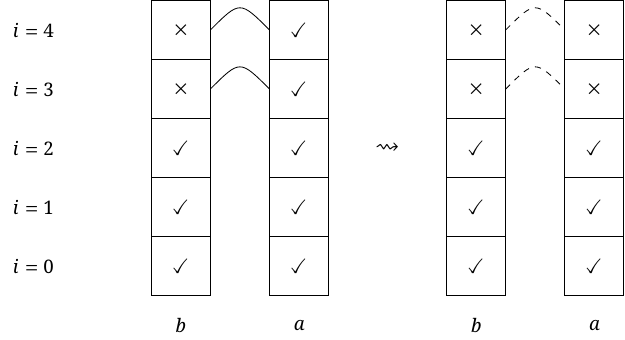}							
		\caption{Construction of the new received word}
		\label{fig:pairing transform}
	\end{cfigure}
	
	\begin{proof}
		Let $(a, b)$ be one of the pairs in the pairing. 
		We count the contribution of $a$ and $b$ to $\Gamma$ before and after the transformation is applied.
		
		Suppose $i_R(a) = i_a$ (originally) and $i_R(b) = i_b$. Since it is a good pairing, $i_a > i_b$ and $\omega(b) \leq \omega (a)$. 
		
		Originally, their contribution to $\Gamma_w^{d, \ell, s} (g, R)$ is at least \[\omega(a) + ((s-i_b)n - (d - \ell) - w(b, i_b)).\]
		
		The reason why the contribution from $a$ is at least $\omega(a)$ is as follows. Let $\omega(a) = w(a, i^*)$. That is, let $i^* = \argmax_{i \in [r]} {w(a, i)}$. If $i^* < i_a$, then $\max_{j<i_a} w(a, i) = \omega(a)$. On the other hand, if $i^* \geq i_a$, we have that $w(a, i^*) \leq \frac{1}{2}(n(s-i^*)-(d-\ell))$, as well as  $w(a, i_a) \leq \frac{1}{2}(n(s-i_a)-(d-\ell))$. Then the contribution to $\Gamma$ is at least $ n(s-i_a)-(d-\ell) - w(a, i_a) \geq n(s-i^*)-(d-\ell) - w(a, i^*) \geq \frac{1}{2}( n(s-i^*)-(d-\ell)) \geq w(a, i^*) = \omega(a)$.
		
		Now, since $w(b, i_b) \leq \omega(b) \leq \omega(a)$, the total contribution is at least $(s-i_b)n - (d - \ell)$.
		
		Now, we apply the corruption to $i_a - i_b$ levels of $a$ and change the weight function $w$ at $a$ and $b$.
		% Note that we ``move'' those entries so they have just enough distance to be wrongly decoded. 
		Then, to $\Gamma_{w'}^{d, \ell, s} (g', R)$, $a$ and $b$ each contribute $\frac{1}{2}((s-i_b)n - (d-\ell))$, giving a total of $((s-i_b)n - (d-\ell))$, which is the lower bound for the previous expression.
		
		Notice, that after we have applied this transformation to all the pairs present in the good pairing we are left with a word $g'$ such that $|A|\leq \ell$. This is because by Condition $4$ of \cref{def: good pairing} we have $\sum_{j=0}^{k-1} (i_R(a_j) - i_R(b_j)) \geq u$: hence, we have corrupted at least $u$ locations $(a,i)$ of agreement between $g$ and $r$ to obtain $g'$. However, now $g'$ and $w'$ satisfy the hypothesis of \autoref{clm: size of A}.
		Putting everything together we get that  $\Gamma_w^{d, \ell, s}(g, R) \geq \Gamma_{w'}^{d, \ell, s}(g, R)\geq \frac{n^2}{2}(s-\frac{d}{n})$.    
	\end{proof}
	
	We can finally complete the proof of \cref{LEM:THRESHOLD_EXISTS}.
	
	\begin{proof}[Proof of Lemma~\ref{LEM:THRESHOLD_EXISTS}]
		We prove this by contradiction, by assuming that no threshold $\thetabar$ is good. By \autoref{obs:decodable}, this means we assume $\abs{A_{\thetabar}} \leq \abs{B_{\thetabar}} + \ell$ for all $\thetabar$.

		\autoref{clm: size of A} shows that $A$ must have size at least $\ell$, or the error is more than promised. Hence, we can write $\abs{A} = \ell + u$ for some positive $u$. 
		
		As in our reproof of Forney's result in \cref{subsec: alternate_Forney}, we construct a pairing. Here, the size of the pairing is $u$ - measured not as the number of locations $a\in T$ participating in the pairing  but the total difference in the multiplicity levels summed over each pair. We add some conditions for it to be a \emph{good} pairing in \autoref{def: good pairing}: \autoref{LEM:GOOD_PAIRING_EXISTS} shows that such a pairing exists. 
		Given such a \emph{good} pairing \autoref{clm: error decreases with replacement} shows that  $\Gamma_w^{d, \ell, s}(g, R) \geq \frac{n^2}{2}(s-\frac{d}{n})$, a contradiction.
	\end{proof}

	\section{Bivariate multiplicity code decoder}
	\label{sec:bivariate}
	
	In this section, we describe our decoding
	algorithm for bivariate multiplicity codes. It relies on
	algorithms for decoding two variants of a decoder for univariate
	multiplicity codes; the first where the multiplicity parameters varies
	with the evaluation point (\cref{alg:gen mult decoder}) and the second being a decoder for a
	weighted version of univariate multiplicity codes (\cref{alg:weighted multiplicity code decoder}). 
	
	\subsection{Description of the bivariate decoder}\label{sec:bivariate_decoder}
	
	We start with an informal description of this algorithm, along the
	lines suggested in the overview \cref{sec: overview}. 
	
	The bivariate decoder takes as input sets $T_1, T_2 \subseteq \F$ of size
	$n$, degree and multiplicity parameters $d$ and $s$, and the received
	word $f \colon T_1 \times T_2 \to \mathbb{F}_{<s}[z_1, z_2]$. The
	decoder outputs a polynomial $P \in \F_{\leq d}[x_1,x_2]$ such that
	$\Delm^{(s)}(f,\enc^{(s)}(P)) < \frac{1}{2} n^2 (s-\frac{d}{n})$ if one
	exists. It will be convenient to write the polynomial $P$ in the
	following form.
	\[P(x_1,x_2) = \sum_{\ell \in [d+1]} P_{\ell}(x_1)x_2^{d-\ell} \in \mathbb{F}_{\leq d}[x_1, x_2]\;.\]
	The decoder proceeds in $d+1$ iterations numbered $0$ to $d$ where in
	the $\ell^{th}$ iteration, the decoder recovers the univariate
	polynomial $P_\ell(x_1) \in \F_{\leq \ell}[x_1]$. Having obtained the
	polynomials $P_0,P_1,\ldots, P_{\ell-1}$ correctly in the previous
	iterations, the decoder in the $\ell^{th}$ iteration peels away these
	polynomials from the received word $f$ to obtain the word $f_\ell$
	defined as \[f_\ell :=   f - \enc^{(s)}\left(\sum_{i \in [\ell]} P_i(x_1)
	x_2^{d-i}\right). \] The appropriate choice of the multiplicity parameter
	for the $\ell^{th}$ iteration is $r = s - \lfloor \frac{d-\ell}{n}
	\rfloor$ (see \cref{rem:multr}). For each $a \in T_1$, the decoder
	unravels the received word $f_\ell|_{x_1 = a}$ along the column $x_1 = a$ into
	$r$ parts $\left\{f_{\ell}^{(i, a)} \right\}_{i \in [r]}$ in
	\cref{line:fellparts}. It then runs the univariate decoder on each
	of these parts to obtain the polynomials $G^{(i,a)}_\ell(x_2)
	\in \F_{d-\ell}[x_2]$ respectively for each $i \in [r]$. It
	then constructs the polynomial $g_\ell$ from the leading
	coefficients of these polynomials (\cref{line:defn_g_ell}) and the
	corresponding weights $w_\ell(a,i)$ which indicate how
	close the polynomial $G^{(i,a)}_\ell(x_2)$ is to the word
	$f_\ell^{(i,a)}$ (\cref{line:defn_w_ell}). Finally it
	extracts the polynomial $P_\ell$ from the pair $(g_\ell,
	w_\ell)$ using the Weighted Univariate Multiplicity Code Decoder.\\
	
	\setcounter{AlgoLine}{0}
	\begin{algorithm}%[H]
		\caption{Bivariate Multiplicity Code Decoder}
		\label{alg:outer}
		\nonl 	\KwIn{$T_1, T_2 \subseteq \F$, $|T_1| = |T_2| = n$ \Comment*[f]{points of evaluation}}
				\myinput{$d, s$ \Comment*[f]{degree and multiplicity resp.}}
				\myinput{$f \colon T_1 \times T_2 \to \mathbb{F}_{<s}[z_1, z_2]$.  \Comment*[f]{received word}}
		\KwOut{$P = \sum_{\ell \in [d+1]} P_{\ell}\left(x_1\right)x_2^{d-\ell} \in \mathbb{F}_{\leq d}[x_1, x_2]$ such that $\Delm^{(s)}\left(f,\enc^{(s)}(P)\right) < \frac{1}{2} n^2 \left(s-\frac{d}{n}\right)$, if such a $P$ exists}
		\For{$\ell \gets 0$ to $d$}{ \label{line:for_loop_in_outer_algo}
			Define $f_{\ell}\colon T_1 \times T_2 \to \mathbb{F}_{<s}[z_1,
			z_2]$ as \Comment*[r]{received word for the $\ell^{th}$ iteration}
			\nonl\Indp\Indp $f_\ell \gets   f - \enc^{(s)}\left(\sum_{i \in [\ell]} P_i\left(x_1\right)
			x_2^{d-i}\right)$\,;\\ 
			\Indm\Indm
			\nonl $\forall \left(a, b\right) \in T_1 \times T_2$, let \Comment*[r]{some
				more notation}
			\nonl\Indp\Indp $f_{\ell}\left(a, b\right) = \sum_{i, j}f_{\ell, \left(i,
				j\right)}\left(a, b\right) z_1^{i}z_2^{j}$ \,;\\
			\Indm\Indm
			Set  $r \gets s - \lfloor \frac{d-\ell}{n} \rfloor $\,; \label{line:noder}
			\Comment*[r]{number of usable derivatives}
			
			\For{$i \gets 0$ to $ r-1$}{
				\For{$a \in T_1$}{
					Define $f_{\ell}^{\left(i, a\right)}\colon T_2 \to
					\F_{<s-i}[z_2]$ as\label{line:fellparts}\\
					\nonl\Indp\Indp $f_{\ell}^{\left(i, a\right)} \left(b\right) \gets \sum_{j \in
						[s-i]} f_{\ell, \left(i, j\right)}\left(a, b\right) \cdot z_2^{j}$\,;
					\\
					\Indm\Indm
					Run Generalized Univariate Multiplicity Code Decoder (\cref{alg:gen mult decoder}) on $\left(T_2,  d-\ell, s-i, f_{\ell}^{\left(i, a\right)}\right)$ to obtain $G_\ell^{\left(i, a\right)} \in \mathbb{F}_{\leq d-\ell}[x_2]$\,; \label{line:BWmult}\\
					\lIf {$G_\ell^{\left(i, a\right)}$ is $\bot$}{set $G_\ell^{\left(i,
							a\right)} \gets 0$\,;}
				}
			}
			Define $w_\ell\colon T_1\times [r]\to \Z_{\geq 0}$
			as\\
			\nonl\Indp\Indp  $w_\ell \left(a, i\right) \gets  \min
			\left\{\Delm^{\left(s-i\right)}\left(f_\ell^{\left(i,
				a\right)},
			\enc^{(s-i)}\left(G_\ell^{\left(i,
				a\right)}\right) \right), \frac{n}{2}\cdot \left((s-i)-\frac{d-\ell}{n}\right)\right\}
			$\,; \label{line:defn_w_ell} \\
			\Indm\Indm
			Define $g_\ell\colon T_1 \to \F_{<r}[z_1]$ as \label{line:defn_g_ell}\\
			\nonl\Indp\Indp $g_\ell\left(a\right) \gets  \sum_{i \in [r]}
			\coeff_{x_2^{d-\ell}}\left(G_{\ell}^{\left(i,
				a\right)}\right)z_1^{i}$\,; \\
			\Indm\Indm
			Run Weighted Univariate Multiplicity Code Decoder
			(\cref{alg:weighted multiplicity code decoder}) on $\left(T_1, d,
			s, 2,
			\ell, r, g_\ell, w_\ell\right)$ to get $P_\ell\left(x_1\right)$\,;\label{line:wumd} 
		}
		Set $P(x_1, x_2) \gets \sum_{\ell \in [d+1]} P_{\ell}\left(x_1\right)x_2^{d-\ell}$;\\
		\leIf{$\Delm^{\left(s\right)}\left(f,\enc^{(s)}(P)\right) < \frac{1}{2} n^2 \left(s-\frac{d}{n}\right)$}{
			\KwRet{$P\left(x_1, x_2\right)$}}{\KwRet $\bot$\,.}
		
	\end{algorithm}
	
	\begin{remark}\label{rem:multr}
		We remark that the number of derivatives $r$ we use here (\cref{line:noder}) is
		$s-\lfloor \frac{d-\ell}{n} \rfloor$. This is the value that satisfies
		a few constraints. For the call to \cref{alg:gen mult decoder}
		(\cref{line:BWmult}) to work, it must be the case that $(s-i) n
		-(d-\ell)\geq 0$, for
		every $i$ in $[r]$. In other words, we must have $r -1 \leq s
		-(d-\ell)/n$. For the call to \cref{alg:weighted
			multiplicity code decoder} (\cref{line:wumd}), we must have $rn-\ell
		\geq 0$. Finally, in the proof of correctness of \cref{alg:weighted
			multiplicity code decoder} (specifically, \cref{lem: gamma
			triangle}), we require $rn - \ell \geq sn - d$. (This actually
		subsumes the previous constraint.) $r:=s-\lfloor \frac{d-\ell}{n} \rfloor$ satisfies all these (and is the unique integer that does so unless $\frac{d-\ell}{n}$ is an integer). When $\ell=0$, $r \geq 1$, and when $\ell = d$, $r=s$.
	\end{remark}
	
	\subsection{Analysis of the bivariate decoder}
	
	We now formally state the claim of correctness and the running time of
	\cref{alg:outer}. It asserts that \cref{alg:outer} does indeed decode
	bivariate  multiplicity codes from half their minimum distance on
	arbitrary product sets, hence giving us the bivariate version of \cref{thm:main tech result}.
	
	\begin{theorem}[correctness of bivariate decoder (\cref{alg:outer})]%
		\label{THM:CORRECTNESS_OF_BIVARIATE_DECODER}
		%\begin{restatable}[correctness of bivariate decoder (\cref{alg:outer})]{theorem}{bivariatethem}\label{THM:CORRECTNESS_OF_BIVARIATE_DECODER}\RestateRemark
		%\begin{theorem}\label{THM:CORRECTNESS_OF_BIVARIATE_DECODER}
		Let $d, s, n \in \N$ be such that $d < sn$, $\F$ be any  field and let
		$T_1, T_2$ be arbitrary subsets of $\F$ of size $n$ each. Let $f\colon
		T_1 \times T_2 \to \F_{<s}[z_1, z_2]$ be any function. 
		
		Then, on input $T_1, T_2, s, d, n$ and $f$, \cref{alg:outer} runs in time $(sn)^{s + O(1)}$ and outputs a polynomial $P \in \F[x_1, x_2]$ of total degree at most $d$ such that  
		\[
		\Delm^{(s)}(f,\enc^{(s)}_{T_1 \times T_2}(P)) < \frac{1}{2} n^2 \left(s-\frac{d}{n} \right) \, ,
		\]
		if such a $P$ exists. 
	\end{theorem}
	%\end{restatable}
	
	\cref{alg:outer} makes calls to \cref{alg:gen mult decoder} and \cref{alg:weighted multiplicity code decoder}. We have shown the correctness of both individually. 
	
	To prove correctness of \cref{alg:weighted multiplicity code decoder} one of the ingredients we need is that the call on \cref{line:wumd} always satisfies the promise $\Gamma_{w_\ell}^{d, \ell, s}(g_\ell,
	R) < \frac{n^2}{2}(s-\frac{d}{n})$, which is required for \cref{alg:weighted multiplicity code decoder} to work. We prove this in the lemma below.
	
	\begin{lemma}[Relationship between $\Gamma$ and multiplicity distance]\label{lem: gamma comparison to usual distance}
		Suppose we are in the $\ell^{th}$ iteration of the for-loop at \cref{line:for_loop_in_outer_algo} of \cref{alg:weighted multiplicity code decoder}.
		To recall, we have subsets $T_1, T_2 \subseteq \F$  sets of size $n$ each, and natural numbers $d, \ell, s, r$  with $d \geq \ell$ and $r = s - \lfloor \frac{d-\ell}{n} \rfloor$. 
		Let $\tilde{P}_\ell = \sum_{i = \ell}^d P_i(x_1)x_2^{d-i}$ be the polynomial of degree at most $d$ with $\Delm^{(s)}(f_\ell,\enc^{(s)}_{T_1\times T_2}(\tilde{P}_\ell)) < \frac{1}{2} n^2 (s-\frac{d}{n})$, that is, the ``remaining'' portion of $P$ after $P_0, \dotsc, P_{\ell-1}$ have been peeled off.
		
		If $g_{\ell}\colon  T_1 \to \F_{<r}[z_1]$ and $w_{\ell}\colon  T_1 \times [r] \to \Z_{\geq 0}$ are  as defined in \cref{alg:outer}, i.e., for $a\in T_1$, $g_\ell(a)$ is the guess for $(P_{\ell}(a+z_1) \mod \ip{z_1}^r)$ that comes from using \cref{alg:gen mult decoder} on $f_\ell$ at $x_1=a$ and $w_\ell(a,i)$ represents the confidence we have in the $(i-1)^{th}$ derivative of $P_\ell(a+z_1)$ at $a$ as specified by $g_\ell$,  then 
		\[
		\Gamma_{w}^{s, d, \ell}(g_{\ell}, P_{\ell}) \leq \Delm^{(s)}(f_\ell, \enc^{(s)}_{T_1\times T_2}(\tilde{P}_\ell)) < \frac{1}{2}n^2\left(s - \frac{d}{n}\right) \, .
		\]
	\end{lemma}
	
	\begin{proof}	
		It will be helpful to recall the definition of $\Gamma_{w_\ell}^{s, d, \ell}$ from \cref{def: gamma}:
		\begin{multline*}
		\Gamma_{w_\ell}^{s, d, \ell}\left(g_\ell, P_\ell\right)  := \sum_{i = 0}^{r-1} \sum_{a
			\in A_i\left(g_\ell, P_\ell\right)} \max
		\left\{\left(n\cdot \left(\left(s-i\right) -
		\frac{d-\ell}n\right) - w_\ell\left(a,i\right)\right),
		\max_{j < i} w_\ell(a, j) \right\}\\
		+ \sum_{a \in A_r\left(g_\ell,
			P_\ell\right)} \max_{j < r} w_\ell(a,
		j)\;, 
        \end{multline*}
		where for every $i \in [r+1]$
		\[
		A_i\left(g_\ell,P_\ell\right)=\setdef{a\in T_1}{\max{\setdef{j\in [r+1]}{g_\ell(a)=P_\ell(a+z_1) \mod \langle z_1\rangle^{j}}}=i}.
		\]
		
		We also recall the definition of the weights. In short, for every column $a \in T_1$ and level $i$, $f_{\ell}^{\left(i, a\right)}$ is the received word and $G_\ell^{(i,a)}$ is the result of univariate multiplicity decoding. The weight is their multiplicity distance, capped at $\frac{n}2\cdot((s-i)-\frac{d-\ell}n)$.
		
		Formally, $f_\ell^{(i,a)}\colon T_2 \to \F_{< s-i}[z_2]$ refers to 
		\[f_{\ell}^{\left(i, a\right)} \left(b\right) = \sum_{j \in
			[s-i]} f_{\ell, \left(i, j\right)}\left(a, b\right)
		\cdot z_2^{j},\]
		while $G_\ell^{(i,a)}$ is the output of \cref{alg:gen mult decoder} on
		input $\left(T_2,  d-\ell, s-i, f_{\ell}^{\left(i,
			a\right)}\right)$.
		
		Then, the weights $w_\ell$ are defined as
		\begin{align*}
			w_\ell(a,i) &:= \min
			\left\{\Delm^{\left(s-i\right)}\left(f_\ell^{\left(i,
				a\right)} , \enc^{(s-i)}\left(G_\ell^{\left(i, a\right)}\right) \right), \frac{n}2\cdot\left((s-i)-\frac{d-\ell}n\right)\right\}\;.
		\end{align*}    
		
		Having found $G_\ell^{(i,a)}$ for each $a \in T$ and $i \in [s]$, we extract our guess for $P_{\ell}$ for each $a \in T$, as the function $g_\ell \colon T_1 \to \F_{<r}[z_1]$.
		\[g_\ell\left(a\right) := \sum_{i \in [r]}
		\coeff_{x_2^{d-\ell}}\left(G_{\ell}^{\left(i,
			a\right)}\right)z_1^{i}.\]
		
		It follows from the definition that $A_i\left(g_\ell,P_\ell\right)$
		refers to the set of $a$'s in $T_1$ such that $g_\ell$ and $P_\ell$
		agree up to the $(i-1)^{th}$-derivative at $a$ but not at the
		$i^{th}$-derivative for $i <r$ while $A_r\left(g_\ell,P_\ell\right)$
		refers to the set of $a$'s in $T_1$ such that $g_\ell$ and $P_\ell$
		agree up to the $(r-1)^{th}$-derivative. Further, as mentioned earlier, $w_\ell(a,i)$ represents the confidence we have in the $(i-1)^{th}$ derivative of $P_\ell(a+z_1)$ at $a$ as specified by $g_\ell$.
		
		Also, define the polynomial $\tilde{P}_\ell^{(i,a)}(x_2)\in \F_{\leq d-\ell}[x_2]$ as  
		\[
		\tilde{P}_\ell^{(i,a)}(x_2) \coloneqq \coeff_{z_1^i}\left(\sum_{j=\ell}^{d}P_j\left(a+z_1\right)x_2^{d-j}\right).
		\]
		
		That is, the $i^{th}$ Hasse derivative with respect to $x_1$, of $\tilde{P}_\ell$ at $a$.
		
		Hence, if the received word $f_\ell$ and the polynomial $\tilde{P}_\ell = \sum_{j=\ell}
		^d P_j$ agree completely, then the following three conditions are met
		\begin{itemize}
			\item $f_\ell^{(i,a)}(b)=
			\left(\enc^{(s-i)}\left(\tilde{P}_{\ell}^{(i,a)}\right)\right)(b)=\tilde{P}_{\ell}^{(i,a)}\left(b+z_2\right)
			\mod \langle z_2 \rangle^{(s-i)}, \forall i,a,b$.
			\item The polynomials $G_\ell^{(i,a)}$ and $P_\ell^{(i,a)}$ are
			identical, for all $i,a$.
			\item $g_\ell= \enc^{(r)}(P_\ell)$.
		\end{itemize}    
		%We use $f,g,w, P$ instead of $f_\ell, g_\ell, w_\ell, P_\ell$ below for convenience of notation.
		
		Our ultimate goal is to prove $\Gamma_{w}^{s, d, \ell}(g_{\ell}, P_{\ell}) \leq \Delm^{(s)}(f_\ell, \enc^{(s)}_{T_1\times T_2}(\tilde{P}_\ell))$. We will first show 
		\[ \Gamma_{w}^{s, d, \ell}(g_{\ell}, P_{\ell}) \leq \sum_{a \in T_1} \max_{i \in [s]} \left\{\Delm^{(s-i)} \left(f_{\ell}^{(i, a)}, \enc^{(s-i)}\left(\tilde{P}_{\ell}^{(i, a)}\right)\right)\right\} \]
		
		and then 
		\[ \sum_{a \in T_1} \max_{i \in [s]} \left\{\Delm^{(s-i)} \left(f_{\ell}^{(i, a)}, \enc^{(s-i)}\left(\tilde{P}_{\ell}^{(i, a)}\right)\right)\right\} \leq \Delm^{(s)} \left(f_\ell, \encs\left( \tilde{P}_\ell\right)\right) \]
		
		First, fix some $a \in T_1$ and let $a\in A_i \left(g_\ell, P_\ell\right)$ for some $i\in [r+1]$. Recall that its contribution to $\Gamma_{w}^{s, d, \ell}(g_{\ell}, P_{\ell})$ is at least $n\left((s-i) - \frac{d-\ell}n\right) - w_{\ell}(a,i)$. We will show an upper bound on this quantity.
		
		For every $j\in [r+1]$, since the decoded polynomial $G^{(j, a)}$ is the closest, it is no further than the encoding of $\tilde{P}_{\ell}^{(j, a)}$.
		\[ \Delm^{(s-j)} \left(f_{\ell}^{(j, a)}, \enc^{(s-j)}\left(\tilde{P}_{\ell}^{(j, a)}\right)\right) \geq  \Delm^{(s-j)} \left(f_{\ell}^{(j, a)}, \enc^{(s-j)}\left(G_{\ell}^{(j, a)} \right)\right). \] 
		Further if $i<r$, we can now apply the triangle inequality to get 
		\begin{multline*}\Delm^{(s-i)} \left(f_{\ell}^{(i, a)}, \enc^{(s-i)}\left(\tilde{P}_{\ell}^{(i,
			a)}\right)\right) + \Delm^{(s-i)} \left(f_{\ell}^{(i, a)},
		\enc^{(s-i)}\left(G_{\ell}^{(i, a)}\right)\right) \\
        \geq \Delm^{(s-i)}
		\left(\enc^{(s-i)}\left(\tilde{P}_{\ell}^{(i, a)}\right), \enc^{(s-i)}\left(G_{\ell}^{(i,
			a)}\right)\right)\,.\end{multline*}
		
		Since $i <r$ and
		$a \in A_i(g_\ell, P_\ell)$, we have
		$g_\ell(a)\neq P_\ell(a+z_1) \mod \langle z_1\rangle^{i}$. Hence,
		$\tilde{P}_{\ell}^{(i,a)}\neq G_{\ell}^{(i,a)}$ are two distinct polynomials of degree at
		most $d-\ell$ and by the distance of the $(s-i)^{th}$ order
		multiplicity code, we have
		\[\Delm^{(s-i)} \left(\enc^{(s-i)}\left(\tilde{P}_{\ell}^{(i, a)}\right), \enc^{(s-i)}\left(G_{\ell}^{(i, a)}\right)\right)  \geq n\left((s-i) - \frac{d-\ell}n\right),  \]
		which in turn yields, combining with the triangle inequality from above,
		\[ \Delm^{(s-i)} \left(f_{\ell}^{(i, a)}, \enc^{(s-i)}\left(\tilde{P}_{\ell}^{(i, a)}\right)\right) \geq n\left((s-i) - \frac{d-\ell}n\right) - \Delm^{(s-i)} \left(f_{\ell}^{(i, a)}, \enc^{(s-i)}\left(G_{\ell}^{(i, a)}\right)\right).\]
		
		Recall that \[ w_\ell(a,i) = \min
		\left\{\Delm^{\left(s-i\right)}\left(f_\ell^{\left(i,
			a\right)} , \enc^{(s-i)}\left(G_\ell^{\left(i, a\right)}\right) \right), \frac{n}2\cdot\left((s-i)-\frac{d-\ell}n\right)\right\} \]
		
		First we consider the case that the minimum is attained at the first term. That is, \[\Delm^{(s-i)} \left(f_{\ell}^{(i, a)}, \enc^{(s-i)}\left(G_{\ell}^{(i, a)}\right)\right) \leq \frac{n}2\cdot\left((s-i)-\frac{d-\ell}n\right).\] In that case, 
		
		\[ n\left((s-i) - \frac{d-\ell}n\right) - \Delm^{(s-i)} \left(f_{\ell}^{(i, a)}, \enc^{(s-i)}\left(G_{\ell}^{(i, a)}\right)\right) = n\left((s-i) - \frac{d-\ell}n\right) - w_{\ell}(a,i) \]
		
		Otherwise, if $ \Delm^{(s-i)} \left(f_{\ell}^{(i, a)}, \enc^{(s-i)}\left(G_{\ell}^{(i, a)}\right)\right)>  \frac{n}2\cdot\left((s-i)-\frac{d-\ell}n\right)$, then, 
		\[ \Delm^{(s-i)} \left(f_{\ell}^{(i, a)}, \enc^{(s-i)}\left(\tilde{P}_{\ell}^{(i, a)}\right)\right)>  \frac{n}2\cdot\left((s-i)-\frac{d-\ell}n\right). \]
		
		In both cases, we get that if $i<r$, then
		\[ \Delm^{(s-i)} \left(f_{\ell}^{(i, a)}, \enc^{(s-i)}\left(\tilde{P}_{\ell}^{(i, a)}\right)\right) \geq n\left((s-i) - \frac{d-\ell}n\right) - w_{\ell}(a,i).\]
		
		Now, by the definition of $\Gamma$ and $w$, 
		\[ 
		\begin{split}
			\Gamma_w^{s, d, \ell}\left(g_\ell, P_\ell\right) & = \sum_{i = 0}^{r-1} \sum_{a \in A_i\left(g_\ell, P_\ell\right)} \max \left\{\left(n\left((s-i) - \frac{d-\ell}n\right) - w_{\ell}(a,i)\right), \max_{j < i} w_{\ell}(a, j) \right\}\\
            & \qquad \qquad + \sum_{a \in A_r\left(g_\ell, P_\ell\right)} \max_{j < r} w_{\ell}(a, j)\\
			& \leq \sum_{a \in T_1} \max_{i \in [s]} \left\{\Delm^{(s-i)} \left(f_{\ell}^{(i, a)}, \enc^{(s-i)}\left(\tilde{P}_{\ell}^{(i, a)}\right)\right)\right\}\,. 
		\end{split}
		\] 
		
		This proves the first inequality indicated above. We now prove the second:
		% \[\sum_{a \in T_1} \max_{i \in [s]} \Delm^{(s-i)} (f^{(i, a)}, \enc^{(s-i)}(P^{(i, a)})) \leq \Delm^{(s)}(f, P)\;
		%   \]
		% We start with the claim that
		\[ 
		\sum_{a \in T_1} \max_{i \in [s]} \left\{\Delm^{(s-i)} \left(f_{\ell}^{(i,a)}, \enc^{(s-i)}\left(\tilde{P}_{\ell}^{(i,a)}\right)\right)  \right\} \leq \Delm^{(s)} \left(f_\ell, \encs\left( \tilde{P}_\ell\right)\right).
		\]
		To see this, note that the right-hand side is $\sum_{a \in T_1} \sum_{b \in T_2} \left(s-d_{\min}^{(s)}\left(f_\ell(a, b)-\left(\enc^{(s)}\left(\tilde{P}_\ell\right)\right)\left(a, b\right)\right)\right)$ and the left-hand side is $\sum_{a \in T_1} \max_{i \in [s]} \sum_{b \in
			T_2} \left((s-i) - d_{\min}^{(s-i)} \left(f_{\ell}^{(i,a)}(b) - \left(\enc^{(s-i)}\left(\tilde{P}_{\ell}^{(i,a)}\right)\right)\left(b\right)\right)\right)$. Since both
		have a summation over $a \in T_1$, we will show the inequality holds
		term by term for each $a \in T_1$. 
		
		Hence, fix an $a \in T_1$. Let $d_{\min}^{(s)}\left(f_\ell(a, b) -
		\left(\enc^{(s)} \tilde{P}_\ell\right)\left(a, b\right)\right) = d_0$. Then the right-hand side is $s-d_0$.
		
		Say the maximum on the left-hand side is
		attained at some $i_0$. Then the
		left-hand side is \[\sum_{b \in T_2} (s-i_0) -
		d_{\min}^{(s-i_0)}\left(f_{\ell}^{(i_0, a)}(b) - \left(\enc^{(s-i_0)}\tilde{P}_{\ell}^{(i_0,
			a)}\right)(b)\right)\,.\] 
		
		Note that $d_0\leq i_0 +
		d_{\min}^{(s-i_0)}\left(f_\ell^{(i_0, a)}(b) - \left(\enc^{(s-i_0)}\tilde{P}_\ell^{(i_0,
			a)}\right)(b)\right)$, and hence \[s-d_0 \geq (s-i_0) -
		d_{\min}^{(s-i_0)}\left(f_\ell^{(i_0, a)}(b) - \left(\enc^{(s-i_0)}\tilde{P}_\ell^{(i_0,
			a)}\right)(b)\right)\,.\] 
		% If $i_0$ is less than $d_0$, the $d_{\min}^{s-i}$ term is at least $d_0 - i_0$ and the entire term is at most $s-d_0$, which is the same as the left-hand side. Otherwise, if $i_0 \geq d_0$, the $d_{\min}$ term becomes $i_0$ (or possibly even more), so the entire right-hand side is now $s-i_0$, which is at most $s-d_0$. 
		This completes the proof. \qedhere
	\end{proof}
	
	Now that we have all the necessary ingredients, we complete the proof
	of \cref{THM:CORRECTNESS_OF_BIVARIATE_DECODER}.
	
	% \begin{theorem}[Restatement of \cref{THM:CORRECTNESS_OF_BIVARIATE_DECODER}]
		% Let $d, s, n \in \N$ be such that $d < sn$, $\F$ be any  field and let $T_1, T_2$ be arbitrary subsets of $\F$ of size $n$ each. Let $f\colon T_1 \times T_2 \to \F_{<s}[z_1, z_2]$ be an any function. 
		
		% Then, on input $T_1, T_2, s, d, n$ and $f$, \cref{alg:outer} runs in time $O(n^{s + O(1)})$ and outputs a polynomial $Q \in \F[x_1, x_2]$ of total degree at most $d$ such that  
		% $$
		% \Delm^{(s)}(\enc^{(s)}_{T_1 \times T_2}(Q), f) < \frac{1}{2} n^2 \left(s-\frac{d}{n}\right) \, ,
		% $$
		% if such a $Q$ exists. 
		% \end{theorem}
	\begin{proof}[Proof of Theorem~\ref{THM:CORRECTNESS_OF_BIVARIATE_DECODER}]
		\label{proof:bivariate_decoder}
		We first observe that \cref{alg:outer} never outputs an incorrect answer, since towards the end of the algorithm we always check whether the polynomial $P$ that is the potential output indeed 
		satisfies \[
		\Delm^{(s)}(f,\enc^{(s)}_{T_1 \times T_2}(P)) < \frac{1}{2} n^2 \left(s-\frac{d}{n}\right) \, ,
		\]
		and $P$ is output only if the check passes. In particular, \cref{alg:outer} does not output a polynomial if there is no codeword close enough to the received word. 
		
		Thus, to show the correctness of the algorithm, it suffices to assume that there exists a polynomial $Q$ of degree at most $d$ that is close to the received word $f$ and argue that in this case the polynomial $P$  output by the algorithm equals $Q$. Let $Q(x_1,x_2) = \sum_{\ell = 0}^d Q_{\ell}(x_1)x_2^{d-\ell}$ with the polynomials $Q_{\ell}$ satisfying $\deg(Q_{\ell}) \leq \ell$.  
		
		\cref{alg:outer} proceeds in $d+1$ iterations and we now claim that at the end of iteration $\ell$,  we have correctly recovered $Q_{0}, Q_1, \ldots, Q_{\ell}$.  More formally, we have the following claim. 
		
		\begin{claim}\label{claim: correctness of bivariate algorithm induction}
			Let $\ell$ be any element in $\{0, 1, \ldots, d\}$. Then, at the end of the iteration $\ell$ of the \emph{for} loop in line $2$ of \cref{alg:outer}, we have that the polynomial $P_{\ell}(x_1)$ equals $Q_{\ell}(x_1)$. 
		\end{claim}
		Clearly, the claim proves the correctness of \cref{alg:outer}. We now prove this claim by a strong induction on $\ell$. The argument for the base case of the algorithm, i.e. $\ell = 0$ is essentially the same as the in the induction step. So, we just sketch the argument for the induction step. 
		
		To this end, we assume that for every $i \in \{0, 1, \ldots, \ell -1\}$, $P_i(x_1) = Q_i(x_1)$ and prove that $P_{\ell}(x_1) = Q_{\ell}(x_1)$. 
		To start with, let \[\Tilde{Q}_{\ell} := Q - \sum_{i = 0}^{\ell -1} Q_i(x_1)x_2^{d-i} \, .\] From the induction hypothesis, note that 
		\[\Tilde{Q}_{\ell} = Q - \sum_{i = 0}^{\ell -1} P_i(x_1)x_2^{d-i} \, .\]
		Thus, from the definition of the function $f_{\ell} \colon  T_1 \times
		T_2 \to \F_{<s}[z_1, z_2]$  in Line 2 of \cref{alg:outer}
		as \[f_{\ell} := f - \enc^{(s)}\left(\sum_{i = 0}^{\ell -1} P_i(x_1)x_2^{d-i}\right) \,
		,\] and the linearity of the encoding map for multiplicity codes, we have
		\begin{align*}
			\Delm^{(s)}\left(f_\ell, \enc^{(s)}_{T_1 \times T_2}\left(\Tilde{Q}_{\ell}\right)\right) &= \Delm^{(s)}\left(f - \enc^{(s)}\left(\sum_{i = 0}^{\ell -1} P_i\left(x_1\right)x_2^{d-i}\right),\enc^{(s)}_{T_1 \times T_2}\left(Q - \sum_{i = 0}^{\ell -1} Q_{i}\left(x_1\right)x_2^{d-i}\right)\right) \\
			&= \Delm^{(s)}\left(f,\enc^{(s)}_{T_1 \times T_2}\left(Q\right)\right) \\
			&< \frac{1}{2} n^2 \left(s-\frac{d}{n}\right) 
		\end{align*}
		With this guarantee in hand, we now proceed with the analysis of the $\ell^{th}$ iteration. 
		
		Now, for every $a \in T_1$, and $i \in \{0, 1, \ldots, r-1\}$ for $r = s - \lfloor \frac{d-\ell}{n} \rfloor$, the function $f_{\ell}^{(i, a)}\colon  T_2 \to \F_{<s-i}[z_2]$ defined as 
		\[
		f_{\ell}^{(i, a)} (b) = \sum_{j \in [s-i]} f_{\ell, (i, j)}(a, b) \cdot z_2^{j}
		\]
		can be viewed as a received word for a univariate multiplicity code
		with multiplicity $(s-i)$ and degree $d - \ell$ on the set $T_2$ of
		evaluation points. Indeed, if the original received word $f$ had no
		errors, and was in fact the encoding of $Q$, then $f_{\ell}^{(i, a)}$
		must be equal to the encoding of the univariate polynomial obtained by
		taking the $i^{th}$ order (Hasse) derivative of $\tilde{Q}_{\ell}$
		with respect to $x_1$ and setting $x_1$ to $a$. Since the degree of $Q_\ell$ in $x_2$ was at most $d-\ell$, the resulting  $G_{\ell}^{(i, a)}$ also has degree at most $d-\ell$.
		
		By combining the output of various calls to \cref{alg:gen mult decoder}, we obtain the function $g_{\ell}$ and the weight function $w_{\ell}$ which together are part of an input to the Weighted Univariate Multiplicity Code Decoder (\cref{alg:weighted multiplicity code decoder}). Now \cref{lem: gamma comparison to usual distance} shows that $\Delm^{(s)} \left(f_\ell, \encs\left( \tilde{P}_\ell\right)\right) < \frac{1}{2}n^2\left(s - \frac{d}{n}\right)$. This is the promise needed to invoke \cref{thm: correctness of wumd}, which states that \cref{alg:weighted multiplicity code decoder} returns $\tilde{P}_\ell$. 
		
		Since $\Delm^{(s)}(f_{\ell},\enc^{(s)}_{T_1 \times T_2}(\Tilde{P}_{\ell}))$ and $\Delm^{(s)}(f_{\ell},\enc^{(s)}_{T_1 \times T_2}(\Tilde{Q}_{\ell}))$ are both upper-bounded by half the minimum distance between polynomials, $\Tilde{P}_{\ell}$ must equal $\Tilde{Q}_{\ell}(x_1)$, which completes the induction step.

		The upper bound on the running time follows immediately from the time complexity of \cref{alg:gen mult decoder} (\cref{thm:correctness of generalized univariate multiplicity code decoder}) and \cref{alg:weighted multiplicity code decoder} (\cref{thm: correctness of wumd}).		
	\end{proof}

	\section{Multivariate multiplicity code decoder}\label{sec: multivariate algorithm}
	
	In this section, we extend the bivariate decoder (\cref{alg:outer})
	constructed in \cref{sec:bivariate} to the multivariate setting with
	$m>2$. The extension to larger $m$ proceeds as suggested by the
	inductive proof of the multiplicity SZ Lemma. If we perform the
	induction following the standard textbook proof of the SZ Lemma (e.g., in~\cite{AroraBarak} and the proof in Kim and Kopparty's work~\cite{KimK2017}), we
	need a ``weighted multivariate multiplicity code decoder''. However,
	we do not even have a weighted version of the bivariate decoder. We get
	around this issue by performing a slightly different proof of the SZ
	Lemma. This alternative proof of the SZ Lemma
	proceeds by viewing the polynomial as an $(m-1)$-variate polynomial
	with the coefficients coming from a univariate polynomial ring
	$\F[x_m]$ instead of as a univariate polynomial in $x_m$ with the
	coefficients coming from the $(m-1)$-variate polynomial ring
	$\F[x_1, \ldots, x_{m-1}]$. We first present this alternative proof of the
	classical SZ Lemma (without multiplicities) in \cref{sec:msz} and in
	the subsequent sections, extend the bivariate decoder
	(\cref{alg:outer}) to the multivariate decoder (\cref{alg:outer
		multivariate}).

	\subsection{Multivariate Schwartz-Zippel Lemma}\label{sec:msz}
	
	\begin{lemma}[Schwartz-Zippel Lemma]\label{lem: SZ alternate}
		Let $P(x_1, x_2, \ldots, x_m) \in \F[\vx]$ be a non-zero $m$ variate polynomial of total degree at most $d$ and let $T_1, T_2, \ldots, T_m$ be subsets of $F$ of size $n$ each. Then, the number of zeroes of $P$ on the product set $T_1 \times T_2 \times \cdots \times T_m$ is at most $dn^{m-1}$.  
	\end{lemma}
	\begin{proof}
		The proof is via induction on $m$ as usual. The base case, where $m = 1$ is clear. For the induction step, we view $P$ as a polynomial in the variables $x_1, x_2, \ldots, x_{m-1}$, with the coefficients being from the polynomial ring $\F[x_m]$. 
		\[
		P(\vx) = \sum_{i = 0}^{\ell} P_i(\vx) \, , 
		\]
		where the polynomial $P_i(\vx)$ is homogeneous and degree $i$ when viewed as a polynomial in the variables $x_1, \ldots, x_{m-1}$ with coefficients in $\F[x_m]$. Note that $\ell$ is \emph{equal} to the total degree of $P$ in $x_1, \ldots, x_{m-1}$ and is at most $d$, and the degree of $P_i$ in the variable $x_m$ is at most $d-i$. Now, for any setting $a_m$ of $x_m$ in $T_m$,  we consider two cases based on whether $P_{\ell}(x_1, x_2, \ldots, x_{m-1}, a_m)$ is zero or non-zero.  
		
		Let $T_m' \subset T_m$ be the set of $a_m \in T_m$ such that $P_{\ell}(x_1, x_2, \ldots, x_{m-1}, a_m)$ is identically zero. Viewing $P_{\ell}(x_1, x_2, \ldots, x_{m-1}, x_m)$ as a univariate polynomial in $x_m$ of degree at most $d-\ell$ with coefficients from $\F(x_1,\dots,x_{m-1})$, we get that $|T_m'| \leq d-\ell$. For every  each $a_m \in T_m'$, the total number of $(a_1, \ldots, a_{m-1}) \in T_1 \times \cdots \times T_{m-1}$ such that $P(a_1, \ldots, a_{m-1}, a_m)$ equals zero is trivially at most $n^{m-1}$. 
		
		On the other hand for every $a_m \in T_m \setminus T_m'$, $P_{\ell}(x_1, \ldots, x_{m-1}, a_m)$ is not identically zero, and thus $P(x_1, x_2, \ldots, x_{m-1}, a_m)$ is a non-zero $(m-1)$ variate polynomial of degree $\ell$. Thus, for each $a_m \in T_m \setminus T_{m}'$ by the induction hypothesis, the total number of $(a_1, \ldots, a_{m-1}) \in T_1 \times \cdots \times T_{m-1}$ such that $P(a_1, \ldots, a_{m-1}, a_m)$ equals zero is at most $\ell \cdot n^{m-2}$. 
		
		Therefore, the total number of zeroes of $P$ on the product set $T_1 \times T_2 \times \cdots \times T_m$ is at most 
		\[
		|T_m'|\cdot n^{m-1} + (n-|T_m'|)\cdot \ell \cdot n^{m-2} \leq (|T_m'| + \ell )n^{m-1} \, ,
		\]
		which is at most $dn^{m-1}$, since $T_m'$ is of size at most $d-\ell$. 
	\end{proof}
	
	We do remark that the way induction is set up in the above proof is different from the way it proceeds in a typical proof of this lemma, where $P$ is viewed as a univariate polynomial in $x_m$ with the coefficients coming from the $(m-1)$-variate polynomial ring $\F[x_1, \ldots, x_{m-1}]$ (as opposed to being viewed as an $(m-1)$-variate polynomial with the coefficients coming from a univariate polynomial ring $\F[x_m]$.). This subtle difference also shows up in the way induction is done in our decoding algorithm for the multivariate case when compared to how Kim-Kopparty proceed in the decoding algorithm for multivariate Reed-Muller codes~\cite{KimK2017}. In fact, it is not clear to us that the results in this paper can be obtained if we set up the induction as in the work of Kim and Kopparty~\cite{KimK2017}. The main technical difficulty is that we do not have an analogue of \cref{alg:outer} when the received word comes with weights.

	\subsection{Multivariate decoder: description}
	
	We now describe our main algorithm for the multivariate case. For this case, the message space consists of $m$-variate degree $d$ polynomials. For our algorithm, it will be helpful to think of the decomposition of such a polynomial $P$ as follows. For brevity, we use $\ve_{-1} = (e_2, e_3, \ldots, e_m)$ to denote an $m-1$ tuple, whose coordinates are indexed from $2$ up to $m$.
	
	\[
	P(\vx) = \sum_{\ell \in [d+1]} \left( \sum_{\substack{\ve_{-1} \in \Z_{\geq 0}^{m-1} \\ |\ve_{-1}|_{1} = d-\ell}} P_{\ell, \ve_{-1}}(x_1)\cdot \prod_{j = 2}^m x_j^{e_j} \right) \, .
	\]
	Since the total degree of $P$ is at most $d$, the degree of the univariate polynomial $P_{\ell, \ve_{-1}}(x_1)$ is at most $\ell$. 
	
	\subsection*{Multivariate Multiplicity Code Decoder}
	
	% \newpage
        \setcounter{AlgoLine}{0}
	\begin{algorithm}
		\caption{Multivariate Multiplicity Code Decoder}\label{alg:outer multivariate}
		\nonl 	\KwIn{$m; T_1, T_2, \ldots, T_m \subseteq \F$, $|T_1| = \cdots = |T_m|
				= n$ \Comment*[f]{\#variables \& points of evaluation}}
		\myinput{$d, s$ \Comment*[f]{degree and multiplicity resp.}}
		\myinput{$f \colon  T_1 \times T_2 \times \cdots \times T_m \to \mathbb{F}_{<s}[z_1, z_2, \ldots, z_m]$.  \Comment*[f]{received word}}
		\KwOut{$P = \sum_{\ell \in [d+1]} \sum_{\ve \in \Z_{\geq 0}^{m-1},
				|\ve|_{1} = d-\ell} P_{\ell, \ve}\left(x_1\right)\cdot \prod_{j = 2}^m
			x_j^{e_j}\in \mathbb{F}_{\leq d}[x_1, x_2, \ldots, x_m]$ such that
			$\Delm^{\left(s\right)}\left(f,\enc^{(s)}\left(P\right)\right) < \frac{1}{2} n^m \left(s-\frac{d}{n}\right)$, if such
			a $P$ exists. 
		}
		\eIf{m = 2}{
			Run Bivariate Multiplicity Code Decoder (\cref{alg:outer}) on $\left(T_1, T_2, d, s, f\right)$ to obtain $P$ \,; 
		}
		{
			Set  $\tilde{T} \gets  T_2 \times \cdots \times T_{m} $\,;\\
			\For{$\ell \gets 0$ to $d$}{
				Define  $f_{\ell}\colon T_1 \times \tilde{T} \to
				\mathbb{F}_{<s}[\vz]$ as \Comment*[r]{received word
					for the $\ell^{th}$ iteration} 
				\nonl \Indp \Indp $f_\ell  \gets   f - \enc^{(s)}\left(\sum_{i \in [\ell]} \left(
				\sum_{\ve_{-1} \in \Z_{\geq 0}^{m-1}, |\ve_{-1}|_{1} = d-i}
				P_{\ell, \ve_{-1}}\left(x_1\right)\cdot \prod_{j = 2}^m x_j^{e_j} \right)
				\right) $\,; \\
				\Indm \Indm
				\nonl $\forall \va \in T_1 \times \tilde{T}$, let $f_{\ell}\left(\va\right) = \sum_{\vi}f_{\ell, \vi}\left(\va\right) \vz^{\vi}$\,;  \Comment*[r]{some more notation}
				Set $r \gets s - \lfloor \frac{d-\ell}{n} \rfloor $\,; \Comment*[r]{number of usable derivatives}
				\For{$a \in T_1$}{
					\For{$i \gets 0$ to $ r-1$}{
						Define $f_{\ell}^{\left(i, a\right)}\colon  \tilde{T} \to
						\F_{<s-i}[z_2, \ldots, z_m]$ as\\
						\nonl \Indp \Indp $f_{\ell}^{\left(i, a\right)} \left(\va_{-1}\right) =
						\sum_{\vi_{-1} \in [s-i]} f_{\ell, \left(i,
							\vi_{-1}\right)}\left(a, \va_{-1}\right) \cdot \vz_{-1}^{\vi_{-1}}$,
						where $\vz_{-1} = \left(z_2, \ldots, z_m\right)$, $\va_{-1} = \left(a_2,
						\ldots, a_m\right)$ and $\vi_{-1} = \left(i_2, \ldots, i_m\right)$ \,;\\
						\Indm\Indm
						(Recursively) run Multivariate Multiplicity Code Decoder
						(\cref{alg:outer multivariate}) on $\left(m-1, \tilde{T},
						d-\ell, s-i, f_{\ell}^{\left(i, a\right)}\right)$ to obtain
						$G_\ell^{\left(i, a\right)} \in \mathbb{F}_{\leq d-\ell}[x_2,
						x_3, \ldots, x_m]$\,; \\
						\lIf {$G_\ell^{\left(i, a\right)}$ is $\bot$}{set $G_\ell^{\left(i,
								a\right)} \gets 0$\,;}
					}
				}
				Define $w_\ell\colon T_1\times [r]\to \Z_{\geq 0}$
				as\\
				\nonl \Indp \Indp $w_\ell \left(a, i\right)  \gets \min
				\left\{\Delm^{\left(s-i\right)}\left(f_\ell^{\left(i,
					a\right)} , \enc^{(s-i)}\left(G_\ell^{\left(i, a\right)}\right)
				\right),
				\frac{n^{m-1}}{2} \cdot \left( (s-i)-\frac{d-\ell}{n}\right)\right\} $\,;
				\\
				\Indm\Indm
				For every $\ve_{-1} \in \Z_{\geq 0}^{m-1}, |\ve_{-1}|_1 =
				d-\ell$, define $g_{\ell, \ve_{-1}}\colon T_1 \to
				\F_{<r}[z_1]$ as\\
				\nonl \Indp\Indp $g_{\ell, \ve_{-1}}\left(a\right) \gets  \sum_{i
					\in [r]} \coeff_{\vx_{-1}^{\ve_{-1}}}\left(G_{\ell}^{\left(i,
					a\right)}\right)z_1^{i} \,$, 
				where\\
				\nonl $\vx_{-1} = \left(x_2, x_3, \ldots, x_m\right)$, $\ve_{-1} = \left(e_2,
				\ldots, e_m\right)$, and $\vx_{-1}^{\ve_{-1}} = x_2^{e_2}\cdot
				x_3^{e_3}\cdots x_m^{e_m}$\,;\\
				\Indm\Indm
				For every $\ve_{-1} \in \Z_{\geq 0}^{m-1}, |\ve_{-1}|_1 = d-\ell$,
				run Weighted Univariate Multiplicity Code Decoder
				(\cref{alg:weighted multiplicity code decoder}) on $\left(T_1, d, s,m,
				\ell, r, g_{\ell, \ve_{-1}}, w_\ell\right)$ to get $P_{\ell, \ve_{-1}}$\,; 
			} 
		} 
		Set $P(x_1,\dots,x_m) \gets \sum_{\ell \in [d+1]} \sum_{\ve \in \Z_{\geq 0}^{m-1},
			|\ve|_{1} = d-\ell} P_{\ell, \ve}\left(x_1\right)\cdot \prod_{j = 2}^m
		x_j^{e_j}$;\\
		\leIf{$\Delm^{\left(s\right)}\left(f,\enc^{(s)}\left(P\right)\right) < \frac{1}{2} n^m \left(s-\frac{d}{n}\right)$}{
			\KwRet{$P\left(x_1, x_2, \dotsc, x_m\right)$}}{\KwRet{$\bot$}\,.}
	\end{algorithm}

	To prove the correctness of \cref{alg:outer multivariate} we will generalize the analysis of \cref{alg:weighted multiplicity code decoder}. Since, all the proofs are direct extensions of the corresponding proofs in the bivariate case, we have skipped them for the sake of brevity.
	
	\subsection{Weighted Univariate Multiplicity Code Decoder (multivariate case)}
	
	We give the general versions of the statements in \cref{sec:wumd} for the bivariate case. The proofs are identical, except for the new bounds on distances and weights. Accordingly, $n^2$ in the distance statements is replaced by $n^m$, and $n$ in the weight bounds after column decoding is replaced by $n^{m-1}$ (since a column is now replaced by recursive decoding over a grid of lower dimension). There is no change in the matching argument. 
	
	\subsubsection{Properties of weighted distance}
	
	The following is the general definition of the distance $\Gamma$
	between a polynomial and a received word. The only difference from
	\cref{def: gamma} is in the weight bounds ($n$ is replaced by
	$n^{m-1}$) and in the contribution of each element $a$, the quantity $n(s-i)$ is replaced by $n^{m-1}(s-i)$. 
	
	\begin{definition} \label{def:multivariate gamma}
		Let $d, \ell, s, m \in \N$ be parameters with $d \geq \ell$, $T
		\subseteq \F$ be a subset of size $n$. Let $R \in \F[x]$ be a
		univariate polynomial of degree at most $\ell$, $h\colon T \to
		\F_{<r}[z]$ and $w\colon  T \times [r] \to \Z_{\geq 0}$ be functions
		such that for every $(a, i) \in T \times [r]$, $w(a,i) \leq
		\frac{n^{m-1}}{2}\cdot \left((s-i) - \frac{d-\ell}{n}\right)$. 
		
		Then, $\Gamma_w^{s, d, \ell}(h, R)$ is defined as follows. 
		\begin{multline*}
		\Gamma_w^{s, d, \ell}(h, R) =  \left(\sum_{i = 0}^{r-1}\sum_{a \in A_i(h,R)} \max \left\{\left(n^{m-1}\left((s-i) - \frac{d-\ell}{n}\right) - w(a,i)\right), \max_{j < i} w(a, j) \right\}\right)\\
        + \sum_{a\in A_r(h,R)} \max_{j < r} w(a, j)
		\end{multline*}
		where for every $i \in [r+1]$
		\[
		A_i(h,R)=\setdef{a\in T}{\max{\setdef{j\in [r+1]}{h(a)=R(a+z) \mod \langle z\rangle^{j}}}=i}.
		\]
		
	\end{definition}
	
	The following is a triangle-inequality-type statement for the general definition of $\Gamma$. 
	
	\begin{lemma}[Triangle-like inequality for $\Gamma$]
		Let $d, \ell, s, m \in \N$ be parameters with $d \geq \ell$, $T
		\subseteq \F$ be a subset of size $n$. Let $Q, R \in \F[x]$ be
		univariate polynomials of degree at most $\ell$, $h\colon T \to
		\F_{<r}[z]$ and $w\colon  T \times [r] \to \Z_{\geq 0}$ be functions
		such that for every $(a, i) \in T \times [r]$, $w(a,i) \leq
		\frac{n^{m-1}}{2}\cdot \left((s-i) - \frac{d-\ell}n\right)$. 
		
		If $Q \neq R$, then 
		\[
		\Gamma_w^{s,d,\ell}(h, Q) + \Gamma_w^{s,d,\ell}(h, R) \geq n^{m}\left(s - \frac{d}n\right) \, .
		\]
	\end{lemma}
	
	\begin{proof}[Proof sketch]
		The proof is identical to that of \cref{lem: gamma triangle}, with the following one difference: where we count a contribution of the form $(s-i)n - (d-\ell)$, that is replaced by $n^{m-1} ((s-i)-(d-\ell)/n)$ (in accordance with the new weight bound). 
	\end{proof}
	
	The following is the relationship between the distance we work with, $\Gamma$, and the multiplicity distance $\Delta$. 
	
	\begin{lemma}[Relationship between $\Gamma$ and multiplicity distance]
		Let $T_1, T_2, \dotsc, T_m \subseteq \F$ be sets of size $n$ each, and $d, \ell, s, r$ be natural numbers with $d \geq \ell$ and $r = s - \lfloor \frac{d-\ell}{n} \rfloor$. Let 
		\[P = \sum_{i = 0}^d \sum_{\mathbf{x_{-1}} \colon \text{monomials of degree } d-i} P_{i, \mathbf{x_{-1}}}(x_1) \mathbf{x_{-1}} \] 
		be a polynomial of degree at most $d$ with $\Delm^{(s)}(\enc^{(s)}_{T_1\times T_2 \times \dotsb \times T_m}(P), f) < \frac{1}{2} n^m (s-\frac{d}{n})$. If $g_{\ell}\colon  T_1 \to \F_{<r}[z]$ and $w_{\ell}\colon  T_1 \times [r] \to \Z_{\geq 0}$ are  as defined in the main multivariate algorithm, then for every fixed monomial $\mathbf{x_{-1}}$ of degree $d-\ell$,
		\[
		\Gamma_{w}^{s, d, \ell}(g_{\ell}, P_{\ell, x_{-1}}) \leq \Delm^{(s)}(f, \enc^{(s)}_{T_1\times T_2 \times \dotsb \times T_m }(P)) < \frac{1}{2}n^m\left(s - \frac{d}{n}\right) \, .
		\]
	\end{lemma}
	
	\begin{proof}[Proof sketch]
		The proof is identical to that of \cref{lem: gamma comparison to usual distance}, with instances of $(s-i)n$ being once again replaced by $(s-i)n^{m-1}$. 
	\end{proof}
	
	\subsubsection{Proof of \texorpdfstring{\cref{thm:main tech result}}{Theorem 3.9} (correctness of \texorpdfstring{\cref{alg:outer multivariate}}{Algorithm 4})}
	
	Recall that the algorithm proceeds by trying every possible vector of thresholds $\thetabar$. The following lemma asserts that one of the thresholds can be used to carry out the decoding.
	
	\begin{lemma}
		Let $g$ be the received word. If $R$ is such that $\Gamma_w^{d, \ell, s}(g, R) < \frac{n^m}{2}(s-\frac{d}{n})$, then there is a vector of thresholds $\thetabar$ that can be used to find $R$.
	\end{lemma}
	
	\begin{proof}[Proof sketch]
		The proof is identical to that of \cref{LEM:THRESHOLD_EXISTS}:
		it  proceeds by contradiction, assuming that no such vector exists. We start by making a claim on the size of $A$, the set of points where the received word $g$ agrees with the desired polynomial $R$. 
		
		\begin{claim} 
			Let $g$ be any received word and $R$ be a degree $\ell$ polynomial with $\Gamma_w^{d, \ell, s} (g, R) < \frac{n^m}{2}(s-\frac{d}{n})$. Let $A$ be the set of points where $g$ and $R$ agree. Then, $\abs{A} > \ell$. 
		\end{claim}
		
		\begin{subproof}[Proof sketch]
			The proof is identical to that of \autoref{clm: size of A}, with instances of $(s-i)n$ being once again replaced by $(s-i)n^{m-1}$. 
		\end{subproof}
		
		We then define a good pairing as before, the same as \cref{def: good pairing}. \cref{LEM:GOOD_PAIRING_EXISTS} and \autoref{clm: matching to good pairing} which assert that such a good pairing exists, remain the same, word for word.
		
		Finally, we use this pairing to compute the error, and show it is more than promised, under the contradiction assumption that no threshold vector works. The statement of \autoref{clm: error decreases with replacement} remains the same. In its proof, in the error calculation, as expected $(s-i)n$ is replaced by $(s-i)n^{m-1}$. 
	\end{proof}
	
	The remainder of the proof is similar to the proof of \cref{THM:CORRECTNESS_OF_BIVARIATE_DECODER} presented in \cref{sec:bivariate}.
	
	\paragraph{Running time:} Let $W(m)$ denote the running time of the $m$-variate decoder. Then, by \cref{THM:CORRECTNESS_OF_BIVARIATE_DECODER}, $W(2)=(sn)^{s+O(1)}$. From the structure of \cref{alg:outer multivariate}, $W(m)$ satisfies the following recurrence:
	\[
	W(m)\leq (sn)^{s+O(1)}\cdot \binom{m+s-1}{s} + (sn)^{2}\cdot W(m-1).
	\]
	From this we conclude that $W(m)$ is upper bounded by $(sn)^{3m+s+O(1)}\cdot \binom{m+s-1}{s}$.
	
	\section*{Acknowledgements}
	We are thankful to Swastik Kopparty,  Ramprasad Saptharishi and Madhu Sudan for many helpful discussions and much encouragement. 
	A special thanks to Madhu for sitting through many long presentations of the preliminary versions of the proofs and his comments.  

    \newpage
	\printbibliography
\end{document}